\documentclass{lmcs} 
\pdfoutput=1

\usepackage{silence}
\usepackage{lastpage}
\lmcsdoi{16}{1}{27}
\lmcsheading{}{\pageref{LastPage}}{}{}%
{Jan.~30,~2019}{Feb.~28,~2020}{}


\usepackage[T5,T1]{fontenc}
\usepackage[utf8]{inputenc}

\usepackage{tikz}
\usepackage{csquotes}
\usepackage{amssymb}
\usepackage{stmaryrd}
\usepackage{bussproofs}
\usepackage{subcaption}

\usepackage{hyperref}
\usepackage{cleveref} 

\def\eg{\emph{e.g.}}
\def\cf{\emph{cf.}}

\begin{document}

\title[Unique perfect matchings and proof nets]{Unique
  perfect matchings, forbidden transitions
  \texorpdfstring{\\}{ }
  and proof nets for linear logic with Mix}
\titlecomment{{\lsuper*}Extended version of a FSCD 2018 paper}

\author{{\fontencoding{T5}\selectfont Lê Thành Dũng Nguyễn}} 
\address{Université publique, France} 
\email{nltd@nguyentito.eu} 
\thanks{The manifesto
  \url{https://pageperso.lif.univ-mrs.fr/~sylvain.sene/affiliation.html} (in
  French; archived on 2020--02--12 on the Internet Wayback Machine
  (\url{https://archive.org/})) explains the given affiliation.}





\begin{abstract}
  \noindent This paper establishes a bridge between linear logic and mainstream
  graph theory, building on previous work by Retoré (2003). We show that the
  problem of correctness for MLL+Mix proof nets is equivalent to the problem of
  uniqueness of a perfect matching. By applying matching theory, we obtain new
  results for MLL+Mix proof nets: a linear-time correctness criterion, a
  quasi-linear sequentialization algorithm, and a characterization of the
  sub-polynomial complexity of the correctness problem. We also use graph
  algorithms to compute the dependency relation of Bagnol et al.\ (2015) and the
  kingdom ordering of Bellin (1997), and relate them to the notion of blossom
  which is central to combinatorial maximum matching algorithms.

  In this journal version, we have added an explanation of Retoré's
  \enquote{RB-graphs} in terms of a general construction on graphs with
  forbidden transitions. In fact, it is by analyzing RB-graphs that we arrived
  at this construction, and thus obtained a polynomial-time algorithm for
  finding trails avoiding forbidden transitions; the latter is among the
  material covered in another paper by the author focusing on graph theory.
\end{abstract}

\maketitle

\tikzstyle{vertex}=[circle,fill=black,minimum size=7pt,inner sep=0pt]
\tikzstyle{bigvertex}=[circle,draw,thick,fill=black!5,minimum size=16pt,inner sep=0pt]
\tikzstyle{matching edge}=[blue, ultra thick]
\tikzstyle{non matching edge}=[red]
\definecolor{lavenderindigo}{rgb}{0.58, 0.34, 0.92}
\definecolor{amber}{rgb}{1.0, 0.75, 0.0}

\section{Introduction}

\subsection{Algorithmics of proofs in linear logic}

One of the major innovations introduced at the birth of linear
logic~\cite{girard_linear_1987} was a representation of proofs as \emph{graphs},
instead of trees as in natural deduction or sequent calculus. A distinctive
property of these \emph{proof nets} is that checking that a proof is correct
cannot be done merely by a local verification of inference steps: among the
graphs which locally look like proof nets, called \emph{proof structures}, some
are invalid proofs. Hence the \emph{correctness} problem: given a proof
structure, is it a real proof net?

A lot of work has been devoted to this decision problem, and in the case of the
multiplicative fragment of linear logic (MLL), whose proof nets are the most
satisfactory, it can be considered solved from an algorithmic point of view.
Indeed, Guerrini~\cite{guerrini_linear_2011} and Murawski and
Ong~\cite{murawski_fast_2006} have found linear-time tests for MLL correctness;
the problem has also been shown to be $\mathsf{NL}$-complete by Jacobé de
Naurois and Mogbil~\cite{jacobe_de_naurois_correctness_2011}. Both the
linear-time algorithms we mentioned also solve the corresponding search problem:
computing a \emph{sequentialization} of a MLL proof net, i.e., a translation
into sequent calculus.

However, for MLL extended with the \emph{Mix rule}~\cite{fleury_mix_1994}
(MLL+Mix), the precise complexity of deciding correctness has remained unknown
(though a polynomial-time algorithm was given by
Danos~\cite{danos_logique_1990}). Thus, one of our goals in this paper is to
study the following problems:

\begin{prob}[\textsc{MixCorr}]%
  \label{prob-correctness}
  Given a proof structure $\pi$, is it an MLL+Mix proof net?
\end{prob}

\begin{prob}[\textsc{MixSeq}]%
  \label{prob-sequentialization}
  Reconstruct a sequent calculus proof for an MLL+Mix proof net.
\end{prob}

\subsection{Proof nets vs graph theory}

It turns out that a \emph{linear-time} algorithm for
\textsc{MixCorr} follows immediately from already known results\footnote{A
  similar historical remark can be made about correctness for MLL without Mix,
  see Remark~\ref{rem-mll-nomix}.}, see Theorem~\ref{thm-linear-time}.
The key is to use a construction by Retoré~\cite{retore_handsome_1999,
  retore_handsome_2003} to reduce it to the problem of \emph{uniqueness of a
  given perfect matching}, which can be solved in linear
time~\cite{gabow_unique_2001}:

\begin{prob}[\textsc{UniquenessPM}]%
  \label{prob-ugpm}
  Given a graph $G$, together with a \emph{perfect matching} $M$ of $G$, is $M$
  the only perfect matching of $G$? Equivalently, is there no \emph{alternating
    cycle} for $M$?
\end{prob}

This brings us to the central idea of this paper: \emph{from the point of view
  of algorithmics, MLL+Mix proof nets and unique perfect matchings are
  essentially the same thing}. This allows us to apply matching theory to the
study of proof nets, leading to several new results. Indeed, one would expect
graph algorithms to be of use in solving problems on proof structures, since
they are graphs! But for this purpose, a bridge between the theory of proof nets
and mainstream graph theory is needed, whereas previous work on the former
mostly made use of \enquote{homemade} objects such as \emph{paired graphs} (an
exception being Murawski and Ong's use of \emph{dominator trees}). By building
on Retoré's discovery of a connection with perfect matchings, this paper
proposes such a bridge.

Thus, proof structures are revealed to be part of a family of graph-theoretic
objects which admit equivalent (as shown by
Szeider~\cite{szeider_theorems_2004}) \enquote{structure from acyclicity}
properties. In linear logic, the corresponding acyclicity property has been
known for a long time: it is the \emph{Danos--Regnier correctness 
  criterion}~\cite{danos_structure_1989}, a necessary and sufficient condition
for a proof structure to be a proof net. These connections have also inspired
new results concerning other members of this family, not only perfect matchings
but also, \eg, \enquote{edge-colored graphs}; that is the subject of another
paper by the author~\cite{nguyen_constrained_2019}.

Another occurrence of an equivalent \enquote{structure from acyclicity} result,
of historical interest for us, is Retoré's
\enquote{aggregates}~\cite[Chapter~2]{retore_reseaux_1993}\footnote{To be more
  accurate, in the reference given, which is a PhD thesis written in French,
  they are called \enquote{agrégats}. However, the word \enquote{aggregate} is
  indeed the official translation, and appeared in the title of the never
  published note \emph{Graph theory from linear logic: Aggregates} (Preprint 47,
  Équipe de Logique, Université Paris 7). That title is also a good summary for
  what we try to achieve in the present paper and
  in~\cite{nguyen_constrained_2019}.}, an early attempt to define a purely
graph-theoretic counterpart to the theory of MLL+Mix correctness. It turns out
that these aggregates occur naturally in graph theory as a tractable case of the
\enquote{rainbow path problem} as we show in~\cite{nguyen_constrained_2019}.

\subsection{Contributions}

First, we establish our equivalence by giving a translation from graphs equipped
with perfect matchings to proof structures (\Cref{sec:equivalence}) --- Retoré's
pre-existing construction takes care of the converse direction\footnote{This is
  a first difference with the conference version, which did not include Retoré's
translation.}. We also propose
later an alternative to Retoré's translation (\Cref{sec:graphification}), having
better properties with respect to sequentialization; this yields a new
graph-theoretic proof of the \emph{sequentialization theorem}, i.e., the
equivalence between MLL+Mix proof nets and Danos--Regnier acyclic proof 
structures.

\subsubsection{Complexity of problems on proof nets}

As already mentioned, we give the first linear-time algorithm for
\textsc{MixCorr}~(\Cref{sec:linear-time}). As for its sub-polynomial complexity
(\Cref{sec:sub-polynomial}), we show that \textsc{MixCorr} is in randomized
$\mathsf{NC}$ and in $\mathsf{quasiNC}$ (informally, $\mathsf{NC}$ is the class
of problems with efficient \emph{parallel} algorithms). On the other hand, we
have a sort of hardness result: if \textsc{MixCorr} were in $\mathsf{NC}$ --- in
particular, if it were in $\mathsf{NL}$, as for MLL without Mix --- this would
imply a solution to a long-standing conjecture by Lovász
(Conjecture~\ref{upm-conjecture}) concerning the related \emph{unique perfect
  matching} problem:

\begin{prob}[{\textsc{UniquePM}~\cite{kozen_nc_1985, gabow_unique_2001,
      hoang_bipartite_2006}}]%
  \label{prob-upm}
  Given a graph $G$, determine whether it admits exactly one perfect matching
  and, if so, find this matching.
\end{prob}

We then turn to the sequentialization problem, for which we provide a
graph-theoretic reformulation --- thanks to our new translation in
\Cref{sec:graphification} --- and an algorithm relying on this reformulation.
This gives us a \emph{quasi-linear} time\footnote{More precisely, $O(n {(\log
  n)}^2 {(\log \log n)}^2)$ time. Both this and our $\mathsf{quasiNC}$ algorithms
  rely on very recent advances, respectively on dynamic bridge-finding data
  structures~\cite{holm_dynamic_2018} and on the perfect matching existence
  problem~\cite{svensson_matching_2017}. Any further progress on these problems
  would lead to an improvement of our complexity bounds.} solution to
\textsc{MixSeq} (\Cref{sec:sequentialization}); to our knowledge, this beats
previous algorithms for \textsc{MixSeq}.

As a demonstration of our matching-theoretic toolbox, we also show how to
compute some information on the set of \emph{all} sequentializations, namely
Bellin's \emph{kingdom ordering}~\cite{bellin_subnets_1997} of the links of a
MLL+Mix proof net (rediscovered by Bagnol et al.~\cite{bagnol_dependencies_2015}
under the name of \emph{order of introduction}). We give a polynomial time and a
$\mathsf{quasiNC}$ algorithm (\Cref{sec:kingdom-ordering}), both relying on an
effective characterization of this ordering.

\subsubsection{Further connections to graph theory}

We also show that this notion of kingdom ordering admits a direct counterpart in
unique perfect matchings. The above-mentioned characterization, when rephrased
in the language of graph theory (\Cref{sec:blossoms}), turns out to involve
objects which play a major role in matching algorithms, namely
\emph{blossoms}~\cite{edmonds_paths_1965}. In this way, we obtain a new result
of independent interest in combinatorics. The appendix of the conference version
of this paper contained a direct proof of this result; instead of reproducing it
here, we have moved it to the companion paper~\cite{nguyen_constrained_2019},
and limit ourselves here to the equivalence with the already
known~\cite{bellin_subnets_1997} proof net version.

Finally, in \Cref{sec:rb-graph-bis} --- a new section added for this journal
version\footnote{This results of that new section were previously claimed
  without proof in a contributed talk at the 1st International Workshop on
  Trends in Linear Logic and Applications (TLLA 2017).} --- we analyse Retoré's
\enquote{RB-graphs} reduction~\cite{retore_handsome_2003}, and show that it can
be understood in terms of graphs with \emph{forbidden
  transitions}~\cite{szeider_finding_2003} which can be seen as the generalized
paired graphs. This reveals a minor subtlety about what kind of cycles RB-graphs
actually detect in paired graphs.

\tableofcontents

\section{Preliminaries}%
\label{sec:background}

\subsection{Terminology}%
\label{sec:terminology}

\subsubsection{Graph theory}

By default, \enquote{graph} refers to an \emph{undirected} graph. Our
\emph{paths} and \emph{cycles} are \textbf{not allowed to contain repeated
  vertices}\footnote{This choice of terminology is common, see, \eg,~\cite[\S
  1.4]{bang-jensen_digraphs._2009}. The adjective \enquote{elementary} is
  sometimes used to refer to such paths and cycles.}; we will sometimes identify
them with their sets of edges (which characterize them) and apply set operations
on them. A \emph{bridge} of a graph is an edge whose removal increases the
number of connected components.

For directed graphs, the notion of connectedness we consider is \emph{weak
  connectedness}, i.e., connectedness of the graph obtained by forgetting the
edge directions. A \emph{predecessor} (resp.\ \emph{successor}) of a vertex is
the source (resp.\ target) of some incoming (resp.\ outgoing) edge.

\subsubsection{Complexity classes}

We refer to~\cite[\S{}1.4]{jacobe_de_naurois_correctness_2011} for the
logarithmic space classes $\mathsf{L}$ (deterministic) and $\mathsf{NL}$
(non-deterministic) and to~\cite{chandra_constant_1984} for the class
$\mathsf{AC}^0$ of constant-depth circuits. The class $\mathsf{NC}^k$ (resp.\
$\mathsf{quasiNC}^k$~\cite{barrington_quasipolynomial_1992}) consists of the
problems which can be solved by a uniform\footnote{For $\mathsf{NC}^k$ and $\mathsf{quasiNC}^k$,
  we may take this to mean that there is a deterministic logarithmic space
  Turing machine which, given $n$ in unary, computes the circuit for inputs of
  size $n$. We will not enter into the details of $\mathsf{AC}^0$ uniformity.}
family of circuits of depth $O(\log^k n)$ and polynomial (resp.\
quasi-polynomial, i.e., $2^{O(\log^c n)}$) size; $\mathsf{NC} = \bigcup_k
\mathsf{NC}^k$ and $\mathsf{quasiNC} = \bigcup_k \mathsf{quasiNC}^k$.

It is well-known that $\mathsf{AC}^0 \subseteq \mathsf{NC}^1 \subseteq
\mathsf{L} \subseteq \mathsf{NL} \subseteq \mathsf{NC}^2 \subseteq \mathsf{NC}
\subseteq \mathsf{P}$.


\subsection{Perfect matchings, alternating cycles and sequentialization}%
\label{sec:pm}

\begin{defi}
  Let $G = (V,E)$ be a graph. A \emph{matching} (resp.\ \emph{perfect matching})
  $M$ in $G$ is a subset of $E$ such that every vertex in $V$ is incident to
  \emph{at most one} (resp.\ \emph{exactly one}) edge in $M$. An
  \emph{alternating path} (resp.\ \emph{cycle}) for $M$ is a path (resp.\ cycle)
  where, for every pair of consecutive edges, one of them is in the matching and
  the other one is not.
\end{defi}

Testing the existence of a perfect matching in a graph --- or, more generally,
finding a maximum cardinality matching --- is one of the central computational
problems in graph theory. Combinatorial maximum matching algorithms,
starting\footnote{Note that the problem was solved long before in the special
  case of bipartite graphs. In fact, a solution for this case was found in
  Jacobi's posthumous
  papers~\cite{jacobi_investigando_1865,jacobi_looking_2009}.} with Edmonds's
\emph{blossom algorithm}~\cite{edmonds_paths_1965}\footnote{This paper is one of
  the first to propose defining efficient algorithms as polynomial-time
  algorithms; it also contributed to the birth of the field of polyhedral
  combinatorics.}, use alternating paths to iteratively increase the size of the
matching; similarly, alternating cycles are important for the problems
\textsc{UniquenessPM} and \textsc{UniquePM} because they witness the
\emph{non-uniqueness} of perfect matchings.

\begin{lem}[Berge~{\cite{berge_two_1957}}]%
  \label{berge}
  Let $G$ be a graph and $M$ be a perfect matching of $G$. Then if $M' \neq M$
  is a perfect matching, the symmetric difference $M \triangle M'$ is a
  vertex-disjoint union of cycles, which are alternating for both $M$ and $M'$.
  Conversely, if $C$ is an alternating cycle for $M$, then $M \triangle C$ is
  another perfect matching.
\end{lem}

As an example, consider Figure~\ref{fig:pm-non-unique}. The matching on the left
admits an alternating cycle, the outer square; by taking the symmetric
difference between this matching and the set of edges of the cycle, one gets the
matching on the right. Conversely, the symmetric difference between both
matchings (which, in this case, is their union) is the square. Note also that in
Figure~\ref{fig:pm-unique}, there is no alternating cycle because vertex
repetitions are disallowed.

\begin{figure}
  \begin{subfigure}{6.8cm}
  \centering
    \begin{tikzpicture}
      \node[vertex] (w) at (0,2) {};
      \node[vertex] (x) at (2,2) {};
      \node[vertex] (y) at (0,0) {};
      \node[vertex] (z) at (2,0) {};

      \draw [non matching edge] (y) -- (x);
        \draw [matching edge] (w) -- (y);
        \draw [matching edge] (x) -- (z);
        \draw [non matching edge] (w) -- (x);
        \draw [non matching edge] (y) -- (z);

      \node[vertex] (w) at (3,2) {};
      \node[vertex] (x) at (5,2) {};
      \node[vertex] (y) at (3,0) {};
      \node[vertex] (z) at (5,0) {};

      \draw [non matching edge] (y) -- (x);
        \draw [non matching edge] (w) -- (y);
        \draw [non matching edge] (x) -- (z);
        \draw [matching edge] (w) -- (x);
        \draw [matching edge] (y) -- (z);
    \end{tikzpicture}
    \caption{Two PMs of the same graph.}%
    \label{fig:pm-non-unique}
  \end{subfigure}\begin{subfigure}{6.5cm}
    \centering
    \begin{tikzpicture}
      \node[vertex] (w) at (0,2) {};
      \node[vertex] (y) at (0,0) {};
      \node[vertex] (t) at (1.5,1) {};
      \node[vertex] (s) at (3.5,1) {};
      \node[vertex] (x) at (5,2) {};
      \node[vertex] (z) at (5,0) {};

      \draw [matching edge] (w) -- (y);
      \draw [matching edge] (x) -- (z);
      \draw[matching edge] (t) -- (s);
      \draw[non matching edge] (w) -- (t);
      \draw[non matching edge] (y) -- (t);
      \draw[non matching edge] (x) -- (s);
      \draw[non matching edge] (z) -- (s);

    \end{tikzpicture}
    \caption{A graph with a unique PM.}%
    \label{fig:pm-unique}
  \end{subfigure}
  \caption{Examples of perfect matchings (PMs). The edges in the matchings are
    thick and blue.}%
  \label{fig:pm}
\end{figure}
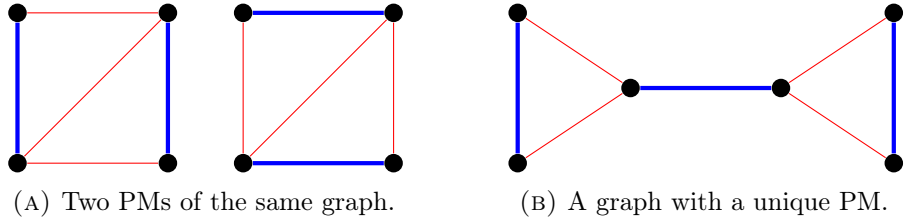

Another approach to finding perfect matchings, using linear algebra, was
initiated by Lovász~\cite{lovasz_determinants_1979} and leads to a
\emph{randomized} $\mathsf{NC}$ algorithm by Mulmuley et
al.~\cite{mulmuley_matching_1987}. Recently, Svensson and Tarnawski have shown
that this algorithm can be derandomized to run in deterministic
$\mathsf{quasiNC}$~\cite{svensson_matching_2017}.

There is also a considerable body of purely mathematical work on matchings,
starting from the 19th century. Let us mention for our purposes a result dating
from 1959.

\begin{thm}[Kotzig~\cite{kotzig_z_1959}]%
  \label{kotzig}
  Let $G$ be a graph. Suppose that $G$ admits a unique perfect matching
  $M$. Then $M$ contains a bridge of $G$.
\end{thm}

As shown by Retoré~\cite{retore_handsome_2003}, Kotzig's theorem leads to an
inductive characterization of the set of graphs equipped with a unique perfect
matching.

\begin{thm}[Sequentialization for unique perfect
  matchings~{\cite{retore_handsome_2003}}]%
  \label{upm-sequentialization}
  The class $\mathcal{UPM}$ of graphs equipped with an unique perfect matching
  is inductively generated as follows:
  \begin{itemize}
  \item The empty graph (with the empty matching) is in $\mathcal{UPM}$.
  \item The disjoint union of two non-empty members of $\mathcal{UPM}$ is in
    $\mathcal{UPM}$.
  \item Let $(G = (V,E), M \subseteq E) \in \mathcal{UPM}$ and $(G' = (V',E'),
    M' \subseteq E') \in \mathcal{UPM}$, with $V$ and $V'$ disjoint. Let $U
    \subseteq V$, $U' \subseteq V'$ such that $U \neq \emptyset$ (resp.~$U' \neq
    \emptyset$) unless $G$ (resp.~$G'$) is the empty graph, and let $x, x'$ be
    two fresh vertices not in $V$ nor $V'$. Then $(G'' = (V'', E''), M''
    \subseteq E'') \in \mathcal{UPM}$, where
    \begin{itemize}
    \item $V'' = V \cup V' \cup \{x,x'\}$
    \item $E'' = E \cup E' \cup \{(x,x')\} \cup (U \times \{x\}) \cup (U' \times
      \{x'\})$
    \item $M'' = M \cup M' \cup \{(x,x')\}$
    \end{itemize}
  \end{itemize}
\end{thm}

\begin{rem}
  By relaxing the non-emptiness condition on $U$ and $U'$, the disjoint union
  operation becomes unnecessary; this is actually the original
  statement~\cite[Theorem 1]{retore_handsome_2003}. Our motivation for this
  change is to get a good fit with Theorem~\ref{graphification-seq-bijection}:
  we want the disjoint union of graphs to correspond to the Mix rule on proof
  nets.
\end{rem}

The inspiration for the above theorem comes from linear logic: it is a
graph-theoretic version of the sequentialization theorems for proof nets, with
Kotzig's theorem being analogous to the \enquote{splitting lemmas} which appear
in various proofs of sequentialization. \Cref{sec:around-seq} is dedicated to
investigating this connection further.

\subsection{Proof structures, proof nets and the correctness criterion}%
\label{sec:proof-structures}

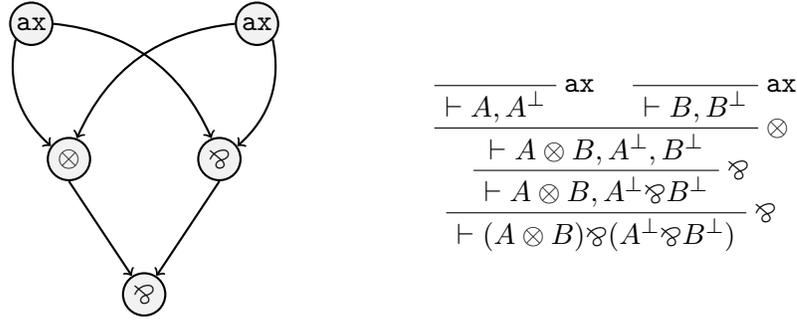
\begin{figure}
  \centering
  \begin{subfigure}{6cm}
    \centering
    \begin{tikzpicture}
      \node[bigvertex] (Ax1) at (1.5,3.6) {$\mathtt{ax}$};
      \node[bigvertex] (Ax2) at (4.5,3.6) {$\mathtt{ax}$};
      \node[bigvertex] (T) at (2,1.8) {$\otimes$};
      \draw[black, thick, ->] (Ax1.south west) to [bend right] (T);
      \draw[black, thick, ->] (Ax2.west) to [bend right] (T);

      \node[bigvertex] (P1) at (4,1.8) {$\bindnasrepma$};
      \draw[black, thick, ->] (Ax1.east) to [bend left] (P1);
      \draw[black, thick, ->] (Ax2.south east) to [bend left] (P1);

      \node[bigvertex] (P2) at (3,0) {$\bindnasrepma$};
      \draw[black, thick, ->] (P1.south) -- (P2);
      \draw[black, thick, ->] (T.south) -- (P2);
  \end{tikzpicture}
    \end{subfigure}\qquad\begin{subfigure}{5cm}
      \centering
      \begin{prooftree}
        \AxiomC{}
        \RightLabel{$\mathtt{ax}$}
        \UnaryInfC{$\vdash A, A^\perp$}
        \AxiomC{}
        \RightLabel{$\mathtt{ax}$}
        \UnaryInfC{$\vdash B, B^\perp$}
        \RightLabel{$\otimes$}
        \BinaryInfC{$\vdash A \otimes B, A^\perp, B^\perp$}
        \RightLabel{$\bindnasrepma$}
        \UnaryInfC{$\vdash A \otimes B, A^\perp \bindnasrepma B^\perp$}
        \RightLabel{$\bindnasrepma$}
        \UnaryInfC{$\vdash (A \otimes B) \bindnasrepma (A^\perp \bindnasrepma B^\perp)$}
      \end{prooftree}
    \end{subfigure}
    \caption{A proof net (left) and its sequentialization (right), written as a
      sequent calculus proof. Edges are usually labeled by the MLL formulae
      appearing in the sequentialization; since we focus on the combinatorics of
      proof structures and not on their logical meaning, we omit them here.}%
  \label{fig:proof-structure}
\end{figure}

\begin{figure}
  \centering
  \[ \frac{}{\vdash A, A^\bot} \text{(\texttt{ax}-rule)}
    \quad
    \frac{\vdash \Gamma, A \quad \vdash B, \Delta}{\vdash \Gamma, A \otimes B,
      \Delta}
    \text{($\otimes$-rule)}
    \quad
    \frac{\vdash \Gamma, A, B}{\vdash \Gamma, A \bindnasrepma B}
    \text{($\bindnasrepma$-rule)}
    \quad
    \frac{\vdash \Gamma \quad \vdash \Delta}{\vdash \Gamma, \Delta}
    \text{(Mix rule)}
  \]
  \caption{Rules for the MLL+Mix sequent calculus; note the correspondence with
    Definition~\ref{def-proof-nets}.}%
  \label{fig:sequent}
\end{figure}

A proof structure is some kind of graph-like object with the precise definition
varying in the literature (we will come back to this point in
Remark~\ref{rem-ps-def}). Since our aim is to apply results from graph theory,
it will be helpful to commit to a representation of proof structures as graphs.

We write $\deg^-$ for the indegree and $\deg^+$ for the outdegree of a vertex.

\begin{defi}%
  \label{def-proof-structure}
  A \emph{proof structure} is a non-empty directed acyclic multigraph $(V,A)$
  with a labeling of the vertices $l : V \to \{\mathtt{ax}, \otimes,
  \bindnasrepma\}$ such that, for $v \in V$:
  \begin{itemize}
  \item if $l(v) = \mathtt{ax}$, then $\deg^-(v) = 0$ and $\deg^+(v) \leq 2$,
  \item if $l(v) \in \{\otimes, \bindnasrepma\}$, then $\deg^-(v) = 2$ and
    $\deg^+(v) \leq 1$.
  \end{itemize}
  Vertices of a proof structure will also be called \emph{links}. A
  \emph{terminal link} is a link with outdegree~0. A \emph{sub-proof structure}
  is a vertex-induced subgraph which is a proof structure.
\end{defi}
\begin{rem}
  It is customary to add \enquote{dangling outgoing edges} from the terminal
  links and to consider them to be the \emph{conclusions} of the proof net. See
  Definition~\ref{def-proof-structure-conclusions} and Remark~\ref{rem-ps-def}.
\end{rem}

\begin{defi}%
  \label{def-proof-nets}
  The set of \emph{MLL proof nets} is the subset of proof structures inductively
  generated by the following rules:
  \begin{itemize}
  \item \textbf{$\mathtt{ax}$-rule}: a proof structure with a single
    $\mathtt{ax}$-link is a proof net.
  \item \textbf{$\otimes$-rule}: if $N$ and $N'$ are proof nets, $u$ is a link
    of $N$ and $v$ is a link of $N'$, then taking the disjoint union of $N$ and
    $N'$, adding a new $\otimes$-link $w$, an edge from $u$ to $w$ and an edge
    from $v$ to $w$ gives a proof net, as long as the resulting graph is a proof
    structure (i.e., the degree constraints are satisfied).
  \item \textbf{$\bindnasrepma$-rule}: if $N$ is a proof net and $u,v$ are links
    of $N$, then adding a new $\bindnasrepma$-link $w$, an edge from $u$ to $w$
    and an edge from $v$ to $w$ gives a proof net, with the same proviso as
    above.
  \end{itemize}
  The set of \emph{MLL+Mix proof nets} is inductively generated by the above
  rules together with the \textbf{Mix rule}: if $N$ and $N'$ are proof nets,
  their disjoint union is a proof net.

  A proof structure is said to be \emph{correct} if it is a MLL+Mix proof net.
\end{defi}

\begin{figure}
  \centering
      \begin{tikzpicture}
      \node[bigvertex] (Ax1) at (-7.5,3.6) {$\mathtt{ax}$};
      \node[bigvertex] (Ax2) at (-4.5,3.6) {$\mathtt{ax}$};
      \node[bigvertex] (T) at (-7,1.8) {$\otimes$};
      \draw[black, thick, ->] (Ax1.south west) to [bend right] (T);
      \draw[black, thick, ->] (Ax2.west) to [bend right] (T);

      \node[bigvertex] (Ax1) at (-3,3.6) {$\mathtt{ax}$};
      \node[bigvertex] (Ax2) at (0,3.6) {$\mathtt{ax}$};
      \node[bigvertex] (T) at (-2.5,1.8) {$\otimes$};
      \draw[black, thick, ->] (Ax1.south west) to [bend right] (T);
      \draw[black, thick, ->] (Ax2.west) to [bend right] (T);

      \node[bigvertex] (P1) at (-0.5,1.8) {$\bindnasrepma$};
      \draw[black, thick, ->] (Ax1.east) to [bend left] (P1);
      \draw[black, thick, ->] (Ax2.south east) to [bend left] (P1);

      \node[bigvertex] (Ax1) at (1.5,3.6) {$\mathtt{ax}$};
      \node[bigvertex] (Ax2) at (4.5,3.6) {$\mathtt{ax}$};
      \node[bigvertex] (T) at (2,1.8) {$\otimes$};
      \draw[black, thick, ->] (Ax1.south west) to [bend right] (T);
      \draw[black, thick, ->] (Ax2.west) to [bend right] (T);

      \node[bigvertex] (P1) at (4,1.8) {$\bindnasrepma$};
      \draw[black, thick, ->] (Ax1.east) to [bend left] (P1);
      \draw[black, thick, ->] (Ax2.south east) to [bend left] (P1);

      \node[bigvertex] (P2) at (3,0) {$\bindnasrepma$};
      \draw[black, thick, ->] (P1.south) -- (P2);
      \draw[black, thick, ->] (T.south) -- (P2);
  \end{tikzpicture}
  \caption{Two successive applications of the $\bindnasrepma$-rule to obtain the
    proof net of Figure~\ref{fig:proof-structure} at the end; compare with the two
    bottom inferences of the sequent calculus proof of
    Figure~\ref{fig:proof-structure}. The leftmost proof net can be obtained by
    invoking the $\mathtt{ax}$-rule to create two proof nets, then combining
    them with a $\otimes$-rule.}%
  \label{fig:inductive}
\end{figure}
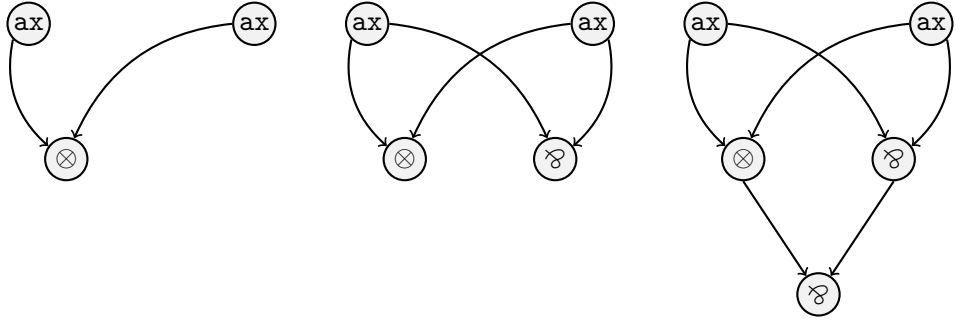

\begin{rem}
  As with any inductively defined set, membership proofs for the set of MLL
  (resp.\ MLL+Mix) proof nets may be presented as inductive derivation trees,
  which are isomorphic to the usual \emph{sequent calculus proofs} of MLL
  (resp.\ MLL+Mix): see Figure~\ref{fig:proof-structure} for an example, and
  Figure~\ref{fig:sequent} for the inference rules of the sequent calculus. An example
  of inductive construction presented directly on proof nets is given in
  Figure~\ref{fig:inductive}.
\end{rem}

\begin{rem}
  The proof structures and proof nets defined here are \emph{cut-free}. This
  restriction is without loss of generality, since a cut link has exactly the
  same behavior as a terminal $\otimes$-link with respect to correctness and
  sequentialization.
\end{rem}

To tackle the problem of correctness, it is useful to have non-inductive
characterizations of proof nets, called \emph{correctness criteria}, at our
disposal. Many of them are formulated using the notion of \emph{paired graphs}.
We will state a criterion first discovered by Danos and Regnier for
MLL~\cite{danos_structure_1989} and extended to MLL+Mix by Fleury and
Retoré~\cite{fleury_mix_1994}.

\begin{defi}%
  \label{def-paired-graph}
  A \emph{paired graph} consists of an undirected graph $G = (V, E)$ and a set
  $\mathcal{P}$ of unordered pairs of edges such that:
  \begin{itemize}
  \item if $\{e, f\} \in \mathcal{P}$, then $e$ and $f$ have a vertex in common;
  \item the pairs are disjoint: if $p, p' \in \mathcal{P}$ and $p \neq p'$, then
    $p \cap p' = \emptyset$.
  \end{itemize}
  When $\{e, f\} \in \mathcal{P}$, the edges $e$ and $f$ are said to be
  \emph{paired}.

  A \emph{switching} $S$ is a set of edges containing exactly one from every
  pair in $\mathcal{P}$. The \emph{switching graph} for $S$ is the spanning
  subgraph $(V, E \setminus(\bigcup{P} \setminus S))$ of $G$, where
  $\bigcup{\mathcal{P}}$ is the union of all pairs in $\mathcal{P}$. A
  \emph{switching path (resp.\ cycle)} is a path (resp.\ cycle) which intersects
  each pair of $\mathcal{P}$ at most once.
\end{defi}

\begin{rem}
  Equivalently, switching cycles are cycles which exist in some switching graph.
\end{rem}

\begin{defi}%
  \label{def-correctness-graph}
  Let $\pi$ be a proof structure. Its \emph{correctness graph} $C(\pi)$ is the
  paired graph obtained by forgetting the directions of the edges and the labels
  of the vertices in $\pi$, and pairing together two edges when their
  targets\footnote{That is, the targets of the directed edges in $\pi$ they come
    from.} are the same $\bindnasrepma$-link.

  A \emph{switching path (resp.\ cycle)} in $\pi$ is a sequence of edges of
  $\pi$ whose image in $C(\pi)$ is a switching path (resp.\ cycle).
\end{defi}

Examples of switchings graphs of a correctness graph are given in Figure~\ref{fig:switching}.

\begin{thm}[Danos--Regnier correctness criterion]
  \label{danos-regnier}
  A proof structure $\pi$ is a MLL (resp.\ MLL+Mix) proof net if and only if all
  the switching graphs of $C(\pi)$ are trees (resp.\ forests).
\end{thm}
\begin{rem}
  Equivalently, $\pi$ is a MLL+Mix proof net if and only if it contains no
  switching cycle.
\end{rem}

\begin{figure}
  \centering
  \begin{tikzpicture}
      \node[bigvertex] (Ax1) at (1.5,3.6) {$\mathtt{ax}$};
      \node[bigvertex] (Ax2) at (4.5,3.6) {$\mathtt{ax}$};
      \node[bigvertex] (T) at (2,1.8) {$\otimes$};
      \draw[black, thick, ->] (Ax1.south west) to [bend right] (T);
      \draw[black, thick, ->] (Ax2.west) to [bend right] (T);

      \node[bigvertex] (P1) at (4,1.8) {$\bindnasrepma$};
      \draw[black, thick, ->] (Ax2.south east) to [bend left] (P1);

      \node[bigvertex] (P2) at (3,0) {$\bindnasrepma$};
      \draw[black, thick, ->] (T.south) -- (P2);
    \end{tikzpicture}
    \qquad\qquad
    \begin{tikzpicture}
      \node[bigvertex] (Ax1) at (1.5,3.6) {$\mathtt{ax}$};
      \node[bigvertex] (Ax2) at (4.5,3.6) {$\mathtt{ax}$};
      \node[bigvertex] (T) at (2,1.8) {$\otimes$};
      \draw[black, thick, ->] (Ax1.south west) to [bend right] (T);
      \draw[black, thick, ->] (Ax2.west) to [bend right] (T);

      \node[bigvertex] (P1) at (4,1.8) {$\bindnasrepma$};
      \draw[black, thick, ->] (Ax2.south east) to [bend left] (P1);

      \node[bigvertex] (P2) at (3,0) {$\bindnasrepma$};
      \draw[black, thick, ->] (P1.south) -- (P2);
    \end{tikzpicture}

    \caption{Two switching graphs out of four possibilities for the proof structure of
      Figure~\ref{fig:proof-structure}.}%
  \label{fig:switching}
\end{figure}
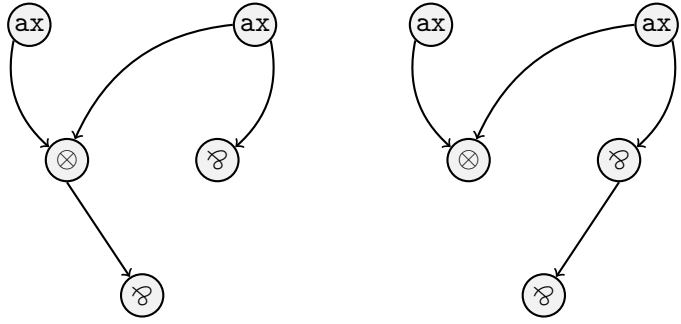

The above is usually called a \emph{sequentialization theorem}: it means that a
proof structure which satisfies the correctness criterion admits a sequent
calculus derivation.

The analogy with Theorem~\ref{upm-sequentialization} is that proof nets are to
proof structures what \emph{unique} perfect matchings are to perfect matchings.
The next section is dedicated to formalizing this analogy into an equivalence.

\section{An equivalence through mutual reductions}%
\label{sec:equivalence}

We will now see how to turn a proof structure into a graph equipped with a
perfect matching, in such a way that switching cycles become alternating cycles,
and vice versa. Such a translation from proof structures to perfect matchings
was first proposed by Retoré~\cite{retore_handsome_2003}, under the name of
\emph{RB-graphs}. After recalling the definition of RB-graphs and their
properties in \Cref{sec:rb-graph}, we propose our own translation in the
converse direction --- which we call the \emph{proofification}
construction --- in \Cref{sec:proofification}.

Thus, this section sets up the reductions (in the sense of complexity theory)
that will be exploited in \Cref{sec:complexity}. But further developments in
\Cref{sec:around-seq} and \Cref{sec:kingdom-ordering} will require the
introduction of a new translation (\emph{graphification}) from proofs to graphs.

\begin{rem}%
  \label{rem-unclear}
  The nature of the object corresponding to a matching edge in a proof structure
  will vary depending on the translation considered: for RB-graphs, they
  correspond to edges or terminal links, whereas in the case of proofifications,
  they are translated into $\otimes$-links. (And in the graphifications of
  \Cref{sec:graphification}, they correspond to links.)

  Thus, by taking the proofification of a RB-graph of a proof structure, one
  gets a different proof structure, with the edges of the former being sent to
  $\otimes$-links of the latter. It is unclear whether this transformation has
  any meaning in terms of linear logic; in particular it does not preserve
  correctness for MLL without Mix.
\end{rem}

\subsection{From proof structures to perfect matchings:
  Retoré's RB-graphs}%
\label{sec:rb-graph}

To define RB-graphs, it is more convenient to start from a slightly altered
definition of proof structures.

\begin{defi}%
  \label{def-proof-structure-conclusions}
  A \emph{proof structure with conclusions} is a non-empty directed acyclic
  multigraph $(V,A)$ with a \emph{partial} labeling of the vertices $l : V
  \rightharpoonup \{\mathtt{ax}, \otimes, \bindnasrepma\}$ such that, for $v \in
  V$:
  \begin{itemize}
  \item if $l(v) = \mathtt{ax}$, then $\deg^-(v) = 0$ and $\deg^+(v) = 2$;
  \item if $l(v) \in \{\otimes, \bindnasrepma\}$, then $\deg^-(v) = 2$ and
    $\deg^+(v) = 1$;
  \item else, $v$ is unlabeled, and then $\deg^-(v) = 1$ and $\deg^+(v) = 0$.
  \end{itemize}
  In the latter case, $v$ is called a \emph{conclusion vertex} and its unique
  incoming edge is called a \emph{conclusion edge}.
\end{defi}

Compared with Definition~\ref{def-proof-structure}, the bounds on the outdegree
have become equalities, while a new kind of vertex has been added. The idea is
that, when the inequality on the outdegree is strict, there are
\enquote{missing} outgoing edges, which are materialized here as conclusion
edges. Yet the object being manipulated is still fundamentally the same; indeed,
the following is immediate (see Figure~\ref{fig:pn-conclusions} for an example):
\begin{prop}%
  \label{pn-conclusions-bijection}
  Given a proof structure with conclusions, the subgraph induced by the labeled
  vertices is a proof structure according to
  Definition~\ref{def-proof-structure}. This correspondence is bijective:
  conversely, there is a unique way to add unlabeled conclusions to a proof
  structure.
\end{prop}

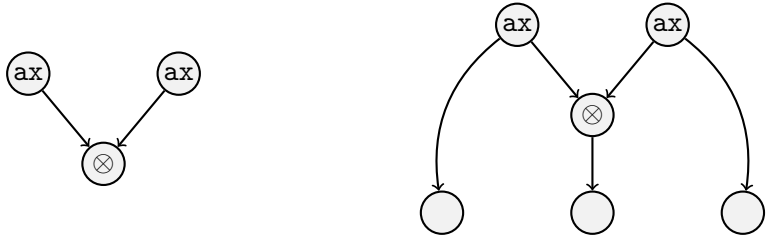
\begin{figure}
  \centering
    \begin{subfigure}{4cm}
    \centering
    \begin{tikzpicture}
    \node[bigvertex] (L) at (1,3.2) {$\mathtt{ax}$};
    \node[bigvertex] (R) at (3,3.2) {$\mathtt{ax}$};
    \node[bigvertex] (N) at (2,2) {$\otimes$};

    \draw[black, thick, ->] (R) -- (N);
    \draw[black, thick, ->] (L) -- (N);
    \end{tikzpicture}
\end{subfigure}\begin{subfigure}{9cm}
  \centering
  \begin{tikzpicture}
  \node[bigvertex] (L) at (1,3.2) {$\mathtt{ax}$};
  \node[bigvertex] (R) at (3,3.2) {$\mathtt{ax}$};
  \node[bigvertex] (N) at (2,2) {$\otimes$};

  \node[bigvertex] (LL) at (0,0.7) {};
  \node[bigvertex] (RR) at (4,0.7) {};
  \node[bigvertex] (NN) at (2,0.7) {};

  \draw[black, thick, ->] (R) -- (N);
  \draw[black, thick, ->] (L) -- (N);
  \draw[black, thick, ->] (N) -- (NN);
  \draw[black, thick, ->] (L) to [bend right] (LL);
  \draw[black, thick, ->] (R) to [bend left] (RR);
  \end{tikzpicture}
\end{subfigure}
  \caption{An instance of the bijection of
    Proposition~\ref{pn-conclusions-bijection}: the proof structure (according
    to Definition~\ref{def-proof-structure}) on the left corresponds to the proof
    structure with conclusions
    (Definition~\ref{def-proof-structure-conclusions}) on the right.}%
  \label{fig:pn-conclusions}
\end{figure}

\begin{rem}%
  \label{rem-ps-def}
  Here we are confronted with the fact that there is no single canonical
  definition of MLL proof structures (although two given definitions are always
  canonically isomorphic). Depending on the task at hand, different
  combinatorial formalizations of the same object may be more or less
  convenient.
  To define proof nets inductively, it was easier to use proof structures
  without conclusions and rely on the notion of terminal link. This will also
  prove useful for the sequentialization algorithm of
  \Cref{sec:sequentialization}. But the conclusion edges are logically
  significant\footnote{An annoying point, however, is that the conclusion
    \emph{vertices} have no significance, so sometimes proof structures are
    defined with \enquote{dangling edges} with no target. However, dangling
    edges drag us out of the world of graphs, and into hypergraphs --- indeed,
    they are hyperedges of arity 1. Proof structures are also often defined as
    the dual hypergraph: links are hyperedges, and formulas are vertices. For
    our purposes, we have chosen to keep proof structures as actual graphs, to
    make the connections with graph theory clearer.}: they correspond to the
  formulas in the sequent being proven.
\end{rem}

Starting from this, we can now introduce RB-graphs. The definition we use is
taken from Straßburger's lecture notes at
ESSLLI'06~\cite{strasburger_proof_2006}, and differs slightly from Retoré's
original one (edges are handled more uniformly in Straßburger's version).

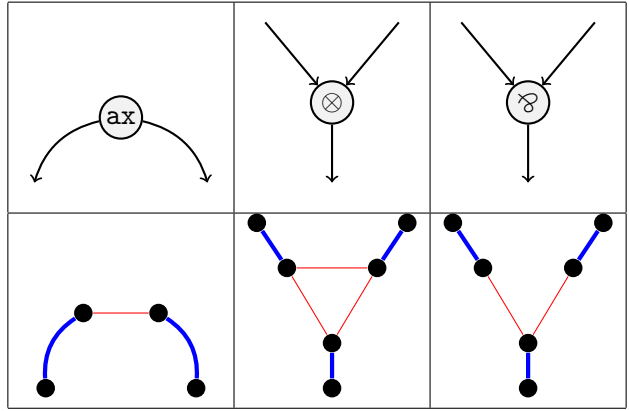
\begin{figure}
  \centering
    \begin{tabular}{|c|c|c|} 
      \hline 
      \begin{tikzpicture}
    \node[bigvertex] (Ax1) at (1.2,1) {$\mathtt{ax}$};
    \node[draw=none] (T) at (0,0) {};
    \node[draw=none] (P1) at (2.4,0) {};

    \draw[black, thick, ->] (Ax1) to [bend left] (P1);
    \draw[black, thick, ->] (Ax1) to [bend right] (T);
      \end{tikzpicture}
      &
      \begin{tikzpicture}
    \node[draw=none] (L) at (0,2.4) {};
    \node[draw=none] (R) at (2,2.4) {};
    \node[bigvertex] (N) at (1,1.2) {$\otimes$};
    \node[draw=none] (C) at (1,0) {};

    \draw[black, thick, ->] (N) -- (C);
    \draw[black, thick, ->] (R) -- (N);
    \draw[black, thick, ->] (L) -- (N);
      \end{tikzpicture}

    &
      \begin{tikzpicture}
    \node[draw=none] (L) at (0,2.4) {};
    \node[draw=none] (R) at (2,2.4) {};
    \node[bigvertex] (N) at (1,1.2) {$\bindnasrepma$};
    \node[draw=none] (C) at (1,0) {};

    \draw[black, thick, ->] (N) -- (C);
    \draw[black, thick, ->] (R) -- (N);
    \draw[black, thick, ->] (L) -- (N);
      \end{tikzpicture}
    \\
      \hline 
      \begin{tikzpicture}
      \node[vertex] (Ax) at (0.5,1) {};
      \node[vertex] (notAx) at (1.5,1) {};

      \node[vertex] (A) at (0,0) {};
      \node[vertex] (notA) at (2,0) {};

      \draw[non matching edge] (Ax) -- (notAx);
      \draw[matching edge] (Ax) to [bend right] (A);
      \draw[matching edge] (notAx) to [bend left] (notA);
      \end{tikzpicture}
    & \begin{tikzpicture}
  \node[vertex] (AtBi) at (10,3) {};
  \node[vertex] (AtBo) at (10.4,2.4) {};
  \draw[matching edge] (AtBi) -- (AtBo);

  \node[vertex] (ApBi) at (12,3) {};
  \node[vertex] (ApBo) at (11.6,2.4) {};
  \draw[matching edge] (ApBi) -- (ApBo);

  \draw[non matching edge] (AtBo) -- (ApBo);

  \node[vertex] (cli) at (11,1.4) {};
  \node[vertex] (clo) at (11,0.8) {};
  \draw[matching edge] (cli) -- (clo);
  \draw[non matching edge] (cli) -- (AtBo);
  \draw[non matching edge] (cli) -- (ApBo);
      \end{tikzpicture}
    & \begin{tikzpicture}
  \node[vertex] (AtBi) at (10,3) {};
  \node[vertex] (AtBo) at (10.4,2.4) {};
  \draw[matching edge] (AtBi) -- (AtBo);

  \node[vertex] (ApBi) at (12,3) {};
  \node[vertex] (ApBo) at (11.6,2.4) {};
  \draw[matching edge] (ApBi) -- (ApBo);

  \node[vertex] (cli) at (11,1.4) {};
  \node[vertex] (clo) at (11,0.8) {};
  \draw[matching edge] (cli) -- (clo);
  \draw[non matching edge] (cli) -- (AtBo);
  \draw[non matching edge] (cli) -- (ApBo);
      \end{tikzpicture}
      \\
    \hline 
  \end{tabular}
  \caption{Translation of proof structures links (top) to RB-graphs
    (bottom).}%
  \label{fig:rb-graph-rules}
\end{figure}

\begin{defiC}[{\cite{retore_handsome_2003,strasburger_proof_2006}}]
  Let $(V,A,l)$ be a proof structure with conclusions. The corresponding
  \emph{RB-graph} is a graph $G$ equipped with a perfect matching $M$ such that:
  \begin{itemize}
  \item $M$ is in bijection with the directed edges $A$;
  \item the non-matching edges of $G$ are derived from the labeling $l : V
    \rightharpoonup \{\mathtt{ax}, \otimes, \bindnasrepma\}$ of the links,
    following the rules of Figure~\ref{fig:rb-graph-rules} (conclusion vertices do not
    induce non-matching edges).
  \end{itemize}
\end{defiC}

\noindent
An example of RB-graph is given in Figure~\ref{fig:rb-graph-example}. The interest of
this translation lies in:
\begin{prop}[implicit in~\cite{retore_handsome_2003}]%
  \label{prop-rb-bijection}
  The switching cycles in a proof structure are in bijection with the alternating
  cycles in its RB-graph.
\end{prop}
\begin{cor}[Retoré's correctness criterion~\cite{retore_handsome_2003}]%
  \label{cor-retore}
  A proof structure satisfies the Danos--Regnier criterion for MLL+Mix if and 
  only if the perfect matching of its RB-graph is unique.
\end{cor}

  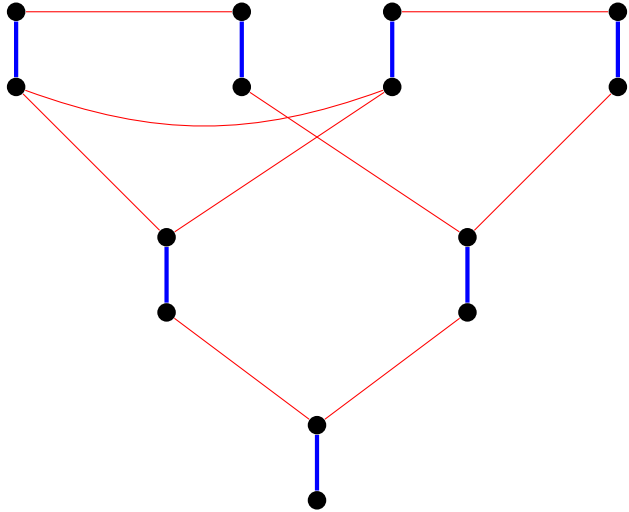
\begin{figure}
    \centering
    \begin{tikzpicture}
      \node[vertex] (A) at (0,6) {};
      \node[vertex] (notA) at (3,6) {};
      \node[vertex] (B) at (5,6) {};
      \node[vertex] (notB) at (8,6) {};

      \node[vertex] (Ax) at (0,7) {};
      \node[vertex] (notAx) at (3,7) {};
      \node[vertex] (Bx) at (5,7) {};
      \node[vertex] (notBx) at (8,7) {};

      \draw[non matching edge] (Ax) -- (notAx);
      \draw[non matching edge] (Bx) -- (notBx);
      \draw[matching edge] (Ax) -- (A);
      \draw[matching edge] (notAx) -- (notA);
      \draw[matching edge] (Bx) -- (B);
      \draw[matching edge] (notBx) -- (notB);

      \node[vertex] (AtBi) at (2,4) {};
      \node[vertex] (AtBo) at (2,3) {};
      \draw[matching edge] (AtBi) -- (AtBo);
      \draw[non matching edge] (AtBi) -- (A);
      \draw[non matching edge] (AtBi) -- (B);
      \draw[non matching edge] (A) to [bend right=20] (B);

      \node[vertex] (ApBi) at (6,4) {};
      \node[vertex] (ApBo) at (6,3) {};
      \draw[matching edge] (ApBi) -- (ApBo);
      \draw[non matching edge] (ApBi) -- (notA);
      \draw[non matching edge] (ApBi) -- (notB);

      \node[vertex] (ccli) at (4,1.5) {};
      \node[vertex] (cclo) at (4,0.5) {};
      \draw[matching edge] (ccli) -- (cclo);
      \draw[non matching edge] (ccli) -- (AtBo);
      \draw[non matching edge] (ccli) -- (ApBo);
    \end{tikzpicture}
    \caption{RB-graph corresponding to the proof net of
      Figure~\ref{fig:proof-structure}.}%
  \label{fig:rb-graph-example}
  \end{figure}

\subsection{From perfect matchings to proof structures}%
\label{sec:proofification}

The translation we present below involves \enquote{$k$-ary
  $\bindnasrepma$-links}. When $k > 1$, these are just binary trees of $k-1$
$\bindnasrepma$-links (correctness is independent of the choice of binary tree:
semantically, this is associativity of $\bindnasrepma$) with $k$ leaves
(incoming edges) and a single root (outgoing edge); the $k=1$ case corresponds
to a single edge and no link.

\begin{defi}
  Let $G = (V,E)$ be a graph and $M$ be a perfect matching of $G$. We define the
  \emph{proofification} of $(G,M)$ as the proof structure $\pi$ built as follows:
  \begin{itemize}
  \item For each non-matching edge $e = (u,v) \in E \setminus M$, we create an
    $\mathtt{ax}$-link $\mathtt{ax}_e$ whose two outgoing edges we will call
    $A_{u,v}$ and $A_{v,u}$.
  \item For each vertex $u \in V$, if $\deg(u) > 1$, we add a $k$-ary
    $\bindnasrepma$-link with $k = \deg(u)-1$, whose incoming edges are the
    $A_{u,v}$ for all neighbors $v$ of $u$ such that $(u,v) \notin M$, and we
    call its outgoing edge $B_u$. If $\deg(u) = 1$, we add an $\mathtt{ax}$-link
    calling one of its outgoing edges $B_u$.
  \item For each matching edge $(u,v) \in M$, we add an $\otimes$-link whose
    incoming edges are $B_u$ and $B_v$. These $\otimes$-links are the terminal
    links of $\pi$.
  \end{itemize}
\end{defi}

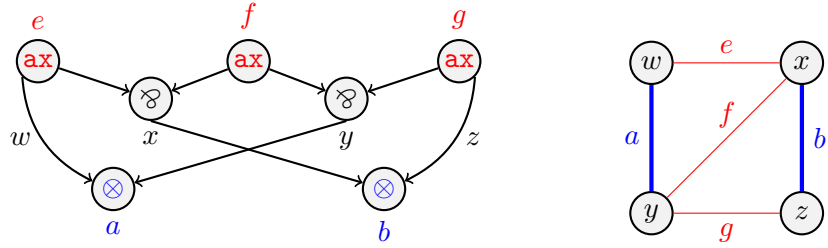
\begin{figure}
  \centering
  \begin{tikzpicture}
    \node[bigvertex, text=red, label=above:{\color{red}$e$}] (Ax1) at (0.5,3.5) {$\mathtt{ax}$};
    \node[bigvertex, text=red, label=above:{\color{red}$f$}] (Ax2) at (3.3,3.5) {$\mathtt{ax}$};
    \node[bigvertex, text=red, label=above:{\color{red}$g$}] (Ax3) at (6.1,3.5) {$\mathtt{ax}$};
    \node[bigvertex, label=below:{$x$}] (P1) at (2,3) {$\bindnasrepma$};
    \node[bigvertex, label=below:{$y$}] (P2) at (4.6,3) {$\bindnasrepma$};
    \node[bigvertex, text=blue, label=below:{\color{blue}$a$}] (T1) at (1.5,1.8) {$\otimes$};
    \node[bigvertex, text=blue, label=below:{\color{blue}$b$}] (T2) at (5.1,1.8) {$\otimes$};

    \draw[black, thick, ->] (Ax1.south west) to [bend right] node[midway, left] {$w$} (T1);
    \draw[black, thick, ->] (Ax3.south east) to [bend left]  node[midway, right] {$z$} (T2);
    \draw[black, thick, ->] (Ax1) -- (P1);
    \draw[black, thick, ->] (Ax2) -- (P1);
    \draw[black, thick, ->] (Ax2) -- (P2);
    \draw[black, thick, ->] (Ax3) -- (P2);
    \draw[black, thick, ->] (P1.south) -- (T2);
    \draw[black, thick, ->] (P2.south) -- (T1);
  \end{tikzpicture}
  \qquad\qquad
  \begin{tikzpicture}
    \node[bigvertex] (w) at (0,2) {$w$};
    \node[bigvertex] (x) at (2,2) {$x$};
    \node[bigvertex] (y) at (0,0) {$y$};
    \node[bigvertex] (z) at (2,0) {$z$};

    \draw [non matching edge] (w) -- node[above] {$e$} ++ (x);
    \draw [non matching edge] (y) -- node[above] {$f$} ++ (x);
    \draw [non matching edge] (y) -- node[below] {$g$} ++ (z);

    \draw [matching edge] (w) -- node[left] {$a$} ++ (y);
    \draw [matching edge] (x) -- node[right] {$b$} ++ (z);
  \end{tikzpicture}
  \caption{The proofification of the graph of Figure~\ref{fig:pm-non-unique}.}%
  \label{fig:proofification-examples}
\end{figure}

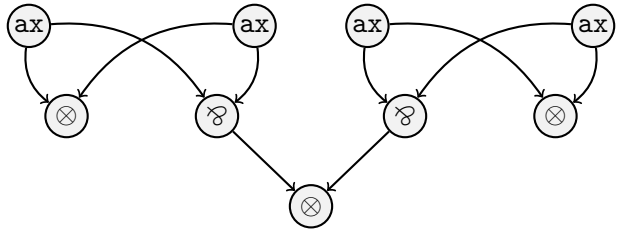
\begin{figure}
  \centering
    \begin{tikzpicture}
    \node[bigvertex] (Ax1) at (1.5,4.4) {$\mathtt{ax}$};
    \node[bigvertex] (Ax2) at (4.5,4.4) {$\mathtt{ax}$};
    \node[bigvertex] (T1) at (2,3.2) {$\otimes$};
    \node[bigvertex] (P1) at (4,3.2) {$\bindnasrepma$};

    \draw[black, thick, ->] (Ax1) to [bend left] (P1);
    \draw[black, thick, ->] (Ax2) to [bend left] (P1);
    \draw[black, thick, ->] (Ax1) to [bend right] (T1);
    \draw[black, thick, ->] (Ax2) to [bend right] (T1);

    \node[bigvertex] (Ax1) at (6,4.4) {$\mathtt{ax}$};
    \node[bigvertex] (Ax2) at (9,4.4) {$\mathtt{ax}$};
    \node[bigvertex] (T2) at (6.5,3.2) {$\bindnasrepma$};
    \node[bigvertex] (P2) at (8.5,3.2) {$\otimes$};

    \draw[black, thick, ->] (Ax1) to [bend left] (P2);
    \draw[black, thick, ->] (Ax2) to [bend left] (P2);
    \draw[black, thick, ->] (Ax1) to [bend right] (T2);
    \draw[black, thick, ->] (Ax2) to [bend right] (T2);

    \node[bigvertex] (X) at (5.25,2) {$\otimes$};
    \draw[black, thick, ->] (P1) -- (X);
    \draw[black, thick, ->] (T2) -- (X);
  \end{tikzpicture}
  \caption{The proofification of the graph in Figure~\ref{fig:pm-unique}. Since the
    perfect matching in Figure~\ref{fig:pm-unique} was unique, we get a MLL+Mix proof
    net. One can check that in this case, it is even correct for MLL.\\
    This proof net will also be used as an example in
    \Cref{sec:kingdom-ordering}.}%
  \label{fig:proofification-ex2}
\end{figure}

\noindent
Examples of proofification are provided in Figure~\ref{fig:proofification-examples}
(annotated figure) and Figure~\ref{fig:proofification-ex2}.

\begin{prop}%
  \label{proofification-cycle}
  Let $G$ be a graph and $M$ be a perfect matching of $G$. The alternating
  cycles for $M$ in $G$ are in bijection with the switching cycles in the proofification
  of $(G,M)$.
\end{prop}
\begin{proof}
  Let $\pi$ be the proofification of $(G,M)$. Any switching cycle in $\pi$ changes
  direction only at $\mathtt{ax}$-links and $\otimes$-links, and therefore can
  be partitioned into an alternation of $\otimes$-links, corresponding to
  matching edges, and of paths starting with some $B_u$, ending with some $B_v$
  and crossing some $\mathtt{ax}_e$, corresponding to non-matching edges $e =
  (u,v)$. Therefore, it corresponds to an alternating cycle for $M$, and the
  mapping defined this way is bijective.
\end{proof}

We end this section on a property of the \emph{sequentializations} of $\pi$.

\begin{prop}%
  \label{proofification-bridge}
  Let $G$ be a graph with a unique perfect matching $M$ and let $\pi$ be the
  proofification of $(G,M)$. A matching edge $e \in M$ is a bridge of $G$ if and
  only if its corresponding $\otimes$-link is introduced by the last rule of
  some sequentialization of $\pi$.
\end{prop}
\begin{proof}
  This follows from the fact that a $\otimes$-link may be introduced by the last
  rule of a sequentialization if and only if it is splitting, i.e., its removal
  disconnects its two precedessors.
\end{proof}

This is consistent with the discussion at the end of \Cref{sec:pm}: a bridge in
a unique perfect matching may be taken as a \enquote{last rule} in its
\enquote{sequentialization} in the sense of Theorem~\ref{upm-sequentialization}.
However, RB-graphs do not satisfy a property analogous to the above proposition.
This is one of the motivations for the introduction of our new translation
(\Cref{sec:graphification}), cf.~Theorem~\ref{graphification-seq-bijection}.



\section{On the complexity of MLL+Mix correctness}%
\label{sec:complexity}

Through the translations of the previous section, MLL+Mix proof \emph{nets}
become \emph{unique} perfect matchings and conversely: these translations
provide \emph{reductions} between the problems \textsc{MixCorr} and
\textsc{UniquenessPM}, allowing us to draw complexity-theoretic conclusions on
proof nets from known results in graph theory. We first look at the time
complexity of \textsc{MixCorr}, then turn to its complexity under constant-depth
($\mathsf{AC}^0$) reductions.

\subsection{An immediate linear-time algorithm}%
\label{sec:linear-time}

Computing Retoré's RB-graphs (\Cref{sec:rb-graph}) and deciding
\textsc{UniquenessPM}~\cite[\S{}3]{gabow_unique_2001} can both be done in linear
time, so:
\begin{thm}%
  \label{thm-linear-time}
  \textsc{MixCorr} can be decided in linear time.
\end{thm}
\begin{rem}%
  \label{rem-mll-nomix}
  By using the \enquote{Euler--Poincaré lemma} (an old part of the linear logic 
  folklore, written down in, \eg,~\cite{bagnol_dependencies_2015}) to
  count the uses of the Mix rule in a proof net, this also allows us to decide
  the correctness of a proof structure for MLL without Mix in linear time.

  Historically, all the necessary ingredients for Theorem~\ref{thm-linear-time}
  and for its corollary for MLL already existed before the announcement (at
  LICS'99, in July 1999) of Guerrini's linear-time correctness criterion for
  MLL~\cite{guerrini_linear_2011}. Indeed, Retoré first presented his RB-graphs
  at the 1996 Linear Logic Tokyo Meeting~\cite{retore_perfect_1996}, and Gabow
  et al.'s algorithm~\cite{gabow_unique_2001} was published at the STOC'99
  conference in May 1999.

  Yet this does not make Guerrini's work obsolete: the latter also gives a way
  to compute a \emph{sequentialization} in linear time for MLL proof nets. The
  other previously known linear-time algorithm for MLL
  correctness~\cite{murawski_fast_2006} also provides a linear-time
  sequentialization procedure. For MLL+Mix, we do not quite manage to match this
  complexity, though we obtain a \emph{quasi-linear} algorithm,
  cf.~\Cref{sec:sequentialization}.

  This is because the methods used in~\cite{guerrini_linear_2011,
    murawski_fast_2006} are quite different from ours: instead of using the
  Danos--Regnier switching acyclicity criterion, their starting points are 
  respectively contractibility (cf.\ Remark~\ref{rem-contractibility}) and
  translation to essential nets (cf.\ \Cref{sec:essential-nets}), which does not
  work with the Mix rule. Therefore, linear time correctness for MLL+Mix is
  absolutely not a trivial generalization of the previous literature on MLL
  without Mix. The discussion at the start of \Cref{sec:kingdom-ordering} makes
  a similar point with respect to sequentialization.
\end{rem}
\begin{rem}%
  \label{rem-alternating-difficult}
  Our decision procedure has the advantage of being simpler to describe than the
  aforementioned algorithms for MLL correctness. That said, this apparent
  simplicity is due to our use of the algorithm of Gabow et
  al.~\cite{gabow_unique_2001} as a black box. Looking inside the black box
  reveals, for instance, that it uses the \emph{incremental tree set union} data
  structure of Gabow and Tarjan~\cite{gabow_linear-time_1985}, which,
  intringuingly, is also a crucial ingredient of
  both~\cite{guerrini_linear_2011, murawski_fast_2006}.

  Finding an alternating cycle is indeed more tricky than in appears at first
  sight. Naively, one would perform a graph traversal which visits alternatively
  matching edges and non-matching edges. The issue is that this would not not
  ensure that the alternating cycle found is \emph{elementary}, i.e., that there
  are no vertex repetitions, which is an essential condition (that we have
  included in our definition of \enquote{cycle} in \Cref{sec:terminology}). The
  difficulty of the problem indeed lies in the interaction of this global
  constraint with the local alternation condition. The analogous issue, seen
  directly on proof structures, is that the traversal does not remember whether
  a premise of a $\bindnasrepma$-link has already been traversed before (in fact
  the standard path-finding algorithms rely on a kind of history independence:
  it does not matter how you reached some intermediate vertex, as long as your
  path was of minimum length).
\end{rem}

\begin{rem}%
  \label{rem-contractibility}
  Gabow et al.'s algorithm for \textsc{UniquenessPM} relies on the technique of
  \emph{blossom shrinking} pioneered by Edmonds~\cite{edmonds_paths_1965}, a
  kind of graph contraction which may remind us of Danos's
  \emph{contractibility} correctness criterion~\cite{danos_logique_1990} for MLL
  without Mix. Indeed, there exists a formal connection: a rewrite step of
  \emph{big-step contractibility}~\cite{bagnol_dependencies_2015} corresponds,
  when translated to either Retoré's RB-graphs or our graphifications
  (\Cref{sec:graphification}), to contracting a blossom. However, not all
  blossoms are redexes for big-step contractibility. See \Cref{sec:blossoms} for
  further discussion of blossoms.
\end{rem}

\subsection{Characterizing the sub-polynomial complexity}%
\label{sec:sub-polynomial}

For MLL proof nets without Mix, correctness is known to be
$\mathsf{NL}$-complete under $\mathsf{AC}^0$ reductions thanks to the
Mogbil--Naurois criterion~\cite{jacobe_de_naurois_correctness_2011}. What about 
MLL+Mix? Since the reductions of \Cref{sec:equivalence} can be computed in
constant depth, we have:

\begin{thm}%
  \label{ugpm-to-mix-correctness}
  \textsc{MixCorr} and \textsc{UniquenessPM} are equivalent under
  $\mathsf{AC}^0$ reductions.
\end{thm}

Thus, it will suffice to study the complexity of \textsc{UniquenessPM}. Let us
start with a positive result, using the parallel algorithms for
perfect matchings mentioned in \Cref{sec:pm}.

\begin{prop}%
  \label{correctness-quasi-nc}
  \textsc{UniquenessPM} is in \emph{randomized} $\mathsf{NC}$ and in
  \emph{deterministic} $\mathsf{quasiNC}$.
\end{prop}
\begin{proof}
  Let $G = (V,E)$ be a graph and $M$ be a perfect matching of $G$. $M$ is
  \emph{not} unique if and only if, for some $e \in M$, the graph $G_e = (V, E
  \setminus \{e\})$ has a perfect matching. To test the uniqueness of $M$, run
  the $|M|$ parallel instances, one for each $G_e$, of a randomized
  $\mathsf{NC}$~\cite{mulmuley_matching_1987} or deterministic
  $\mathsf{quasiNC}$~\cite{svensson_matching_2017} algorithm for deciding the
  existence of a perfect matching, and compute the disjunction of their answers
  in $\mathsf{AC}^0$.
\end{proof}

Being in $\mathsf{quasiNC}$ is a much weaker\footnote{In fact, one can show that
  $\mathsf{NL} \subsetneq \mathsf{NSPACE}(O(\log^{3/2} n)) \subseteq
  \mathsf{quasiNC}^3$, and the problem of finding a perfect matching lies in the
  latter, according to Svensson and Tarnawski's analysis.} result than being in
$\mathsf{NL}$. But as we shall now see, even showing that \textsc{UniquenessPM}
is in $\mathsf{NC}$ (recall that $\mathsf{NL} \subset \mathsf{NC}$) would be a
major result. It would answer in the affirmative the following conjecture dating
back from the 1980's:

\begin{conj}[Lovász\footnote{The conjecture is attributed to Lovász by a
    paper by Kozen et al.~\cite{kozen_nc_1985} which claims to solve it. But
    Hoang et al.~\cite{hoang_bipartite_2006} note that \enquote{this was later
      retracted in a personal communication by the authors}. Still, the proposed
    solution works for bipartite graphs.}]%
  \label{upm-conjecture}
  \textsc{UniquePM} is in $\mathsf{NC}$.
\end{conj}

Indeed, the following shows that $\textsc{UniquenessPM} \in \mathsf{NC}
\Rightarrow \textsc{UniquePM} \in \mathsf{NC}$ (and the converse follows from
the definitions).

\begin{prop}%
  \label{nc2-reduction}
  There is a $\mathsf{NC}^2$ reduction from \textsc{UniquePM} to
  \textsc{UniquenessPM}.
\end{prop}

\begin{proof}
  This is a consequence of a $\mathsf{NC}^2$ algorithm by Rabin and
  Vazirani~\cite[\S4]{rabin_maximum_1989} which, given a graph $G$, computes a
  set of edges $M$ such that if $G$ admits a unique perfect matching, then $M$
  is this matching. Starting from any graph $G$, run this algorithm and test
  whether its output is a perfect matching. If not, then $G$ does not admit a
  unique perfect matching; if it is, then $G$ is a positive instance of
  \textsc{UniquePM} if and only if $(G,M)$ is a positive instance of
  \textsc{UniquenessPM}.
\end{proof}

To sum up these results about \textsc{UniquenessPM}, which apply to
\textsc{MixCorr}:

\begin{thm}
  \textsc{MixCorr} is in randomized $\mathsf{NC}$ and in
  deterministic $\mathsf{quasiNC}$; it is in deterministic
  $\mathsf{NC}$ if and only if Conjecture~\ref{upm-conjecture} is
  true.
\end{thm}

\section{Tackling sequentialization via an appropriate translation}%
\label{sec:around-seq}

We are now interested in using our graph-theoretical tools to deal with
problems concerning the \emph{order of logical rules}, typically that of
computing a \emph{sequentialization} (problem \textsc{MixSeq} from the
introduction). However, there is a mismatch between RB-graphs and proof nets: a
bridge in a RB-graph does not necessarily correspond to the last rule of some
sequentialization of the proof net --- in fact, it generally does not even
correspond to a terminal link. The key issue is indeed that the successor
relation (called $S(\pi)$ in \cref{sec:computing-kingdom}), i.e., \enquote{there
  is a directed edge from $l$ to $l'$} is forgotten by the translation to
RB-graphs.

A toy case to witness the inconvenience caused by this mismatch is the
following: we would like to deduce the sequentialization theorem for the
Danos--Regnier criterion (Theorem~\ref{danos-regnier}) as an immediate corollary 
of Retoré's sequentialization for unique perfect matchings
(Theorem~\ref{upm-sequentialization}). But this is not possible with RB-graphs
-- instead, one must resort to a proof by induction using Kotzig's theorem
(Theorem~\ref{kotzig}), see~\cite[\S{}2.4]{retore_handsome_1999}.

\subsection{A new encoding: graphification}%
\label{sec:graphification}

To fulfill the desiderata mentioned above, we introduce the following
construction, which involves a trick to encode the successor relation.

\begin{defi}
  Let $\pi$ be a proof structure and $L$ be its set of links. The \emph{graphification}
  of $\pi$ is the graph $G = (V,E)$ equipped with a perfect matching $M
  \subseteq E$ with
  \begin{itemize}
  \item the matching edges corresponding to the links: $V = \bigcup_{l \in L}
    \{a_l, b_l\}$, $M = \{(a_l, b_l) \mid l \in L\}$,
  \item and the remaining edges in $E \setminus M$ reflect the incoming edges of
    the $\otimes$-links and $\bindnasrepma$-links, as specified by
    Figure~\ref{fig:graphification-rules}.
  \end{itemize}
\end{defi}

\noindent
Figure~\ref{fig:graphification-example} shows an example of this construction.
As another example, Figure~\ref{fig:pm-unique} from \Cref{sec:pm} is actually the
graphification of Figure~\ref{fig:pn-conclusions} from \Cref{sec:rb-graph}.

\begin{rem}
  There is an ambiguity about the \enquote{$\bindnasrepma$ of $\mathtt{ax}$}
  configuration (cf.~Figure~\ref{fig:par-of-ax}) that can occur in correct proof nets:
  should it result in a multigraph with parallel non-matching edges, or in a
  simple graph? For simplicity we choose the simple graph option, since that is
  the setting for most of the literature on matchings, but this detail has very
  little importance.
\end{rem}

\begin{figure}
  \begin{subfigure}{7.5cm}
    \begin{tabular}{|c|c|} 
      \hline 
      \begin{tikzpicture}
    \node[bigvertex] (L) at (1,3.5) {};
    \node[bigvertex] (R) at (3,3.5) {};
    \node[bigvertex] (N) at (2,2) {$\otimes$};

    \draw[black, thick, ->] (R) -- (N);
    \draw[black, thick, ->] (L) -- (N);
      \end{tikzpicture}

    &
      \begin{tikzpicture}
    \node[bigvertex] (L) at (1,3.5) {};
    \node[bigvertex] (R) at (3,3.5) {};
    \node[bigvertex] (N) at (2,2) {$\bindnasrepma$};

    \draw[black, thick, ->] (R) -- (N);
    \draw[black, thick, ->] (L) -- (N);
      \end{tikzpicture}
    \\
      \hline 
      \begin{tikzpicture}
        \node[vertex] (Al) at (0,2) {};
        \node[vertex] (Ar) at (1,2) {};
        \draw[matching edge] (Al) -- (Ar);

        \node[vertex] (Bl) at (2,2) {};
        \node[vertex] (Br) at (3,2) {};
        \draw[matching edge] (Bl) -- (Br);

        \node[vertex] (Cl) at (1,0) {};
        \node[vertex] (Cr) at (2,0) {};
        \draw[matching edge] (Cl) -- (Cr);

        \draw[non matching edge] (Al) -- (Cl);
        \draw[non matching edge] (Ar) -- (Cl);
        \draw[non matching edge] (Bl) -- (Cr);
        \draw[non matching edge] (Br) -- (Cr);
      \end{tikzpicture}
    & \begin{tikzpicture}
        \node[vertex] (Al) at (0,2) {};
        \node[vertex] (Ar) at (1,2) {};
        \draw[matching edge] (Al) -- (Ar);

        \node[vertex] (Bl) at (2,2) {};
        \node[vertex] (Br) at (3,2) {};
        \draw[matching edge] (Bl) -- (Br);

        \node[vertex] (Cl) at (1,0) {};
        \node[vertex] (Cr) at (2,0) {};
        \draw[matching edge] (Cl) -- (Cr);

        \draw[non matching edge] (Al) -- (Cl);
        \draw[non matching edge] (Ar) -- (Cl);
        \draw[non matching edge] (Bl) -- (Cl);
        \draw[non matching edge] (Br) -- (Cl);
      \end{tikzpicture}
      \\
    \hline 
    \end{tabular}
  \caption{Translation rules for sets of incoming edges.}%
    \label{fig:graphification-rules}
  \end{subfigure}\qquad\begin{subfigure}{5cm}
    \centering
  \begin{tikzpicture}
        \node[vertex] (Ax1l) at (1,4) {};
        \node[vertex] (Ax1r) at (2,4) {};
        \draw[matching edge] (Ax1l) -- node[above, text=black] {$\mathtt{ax}$} ++ (Ax1r);
        \node[vertex] (Ax2l) at (4,4) {};
        \node[vertex] (Ax2r) at (5,4) {};
        \draw[matching edge] (Ax2l) -- node[above, text=black] {$\mathtt{ax}$} ++ (Ax2r);
        \node[vertex] (Tl) at (1.5,2) {};
        \node[vertex] (Tr) at (2.5,2) {};
        \draw[matching edge] (Tl) -- node[above, text=black] {$\otimes$} ++ (Tr);
        \node[vertex] (P1l) at (3.5,2) {};
        \node[vertex] (P1r) at (4.5,2) {};
        \draw[matching edge] (P1l) -- node[above, text=black] {$\bindnasrepma$} ++ (P1r);
        \node[vertex] (P2l) at (2.5,0) {};
        \node[vertex] (P2r) at (3.5,0) {};
        \draw[matching edge] (P2l) -- node[above, text=black] {$\bindnasrepma$} ++ (P2r);

        \draw[non matching edge] (Ax1l) -- (Tl);
        \draw[non matching edge] (Ax1r) -- (Tl);
        \draw[non matching edge] (Ax2l) -- (Tr);
        \draw[non matching edge] (Ax2r) -- (Tr);
        \draw[non matching edge] (Ax1l) -- (P1l);
        \draw[non matching edge] (Ax1r) -- (P1l);
        \draw[non matching edge] (Ax2l) -- (P1l);
        \draw[non matching edge] (Ax2r) -- (P1l);
        \draw[non matching edge] (Tl) -- (P2l);
        \draw[non matching edge] (Tr) -- (P2l);
        \draw[non matching edge] (P1l) -- (P2l);
        \draw[non matching edge] (P1r) -- (P2l);
  \end{tikzpicture}
  \caption{Graphification of the proof structure of Figure~\ref{fig:proof-structure}}%
    \label{fig:graphification-example}
\end{subfigure}
  \caption{The graphification construction.}%
  \label{fig:graphification}
\end{figure}
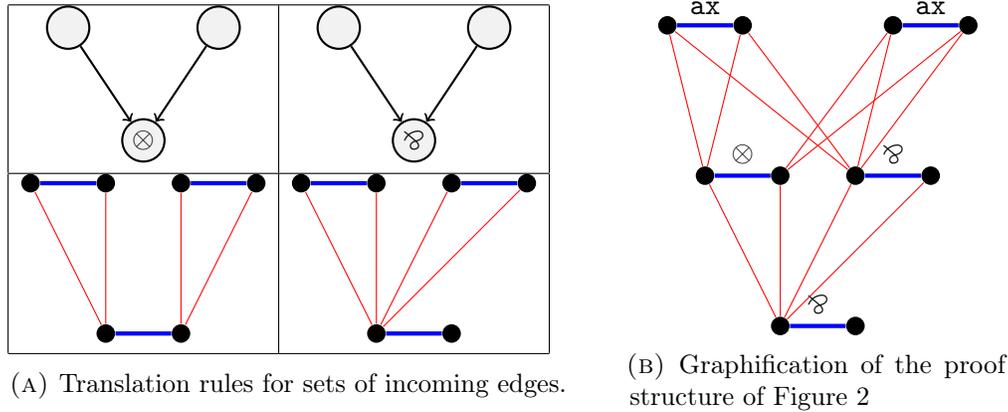

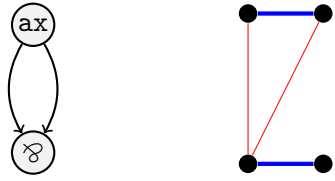
\begin{figure}
  \begin{subfigure}{3cm}
    \begin{tikzpicture}
      \node[bigvertex] (L) at (0,1.7) {$\mathtt{ax}$};
      \node[bigvertex] (N) at (0,0) {$\bindnasrepma$};

      \draw[black, thick, ->] (L) to [bend left] (N);
      \draw[black, thick, ->] (L) to [bend right] (N);
    \end{tikzpicture}
  \end{subfigure}
  \begin{subfigure}{3cm}
    \begin{tikzpicture}
      \node[vertex] (Al) at (1,2) {}; \node[vertex] (Ar) at (2,2) {};
      \draw[matching edge] (Al) -- (Ar);

      \node[vertex] (Cl) at (1,0) {}; \node[vertex] (Cr) at (2,0) {};
      \draw[matching edge] (Cl) -- (Cr);

      \draw[non matching edge] (Al) -- (Cl);
      \draw[non matching edge] (Ar) -- (Cl);
    \end{tikzpicture}
  \end{subfigure}
\caption{The \enquote{$\bindnasrepma$ of $\mathtt{ax}$} configuration and its
  graphification.}%
  \label{fig:par-of-ax}
\end{figure}

\begin{rem}
  To extend Remark~\ref{rem-unclear}, there is no clear relationship between
  graphifications and either of the two translations seen until now (RB-graphs
  and proofifications).
\end{rem}

Just like RB-graphs, graphifications provide a linear time and $\mathsf{AC}^0$
reduction from \textsc{MixCorr} to \textsc{UniquenessPM}: the complexity results
of the previous section could have been obtained using graphifications (as was
done in the conference version). We focus on the soundness of the reduction,
since its complexity is more or less intuitive.
\begin{prop}[Graphification-based correctness criterion]%
  \label{graphification-correctness}
  A proof structure satisfies the Danos--Regnier correctness criterion for 
  MLL+Mix if and only if the perfect matching of its graphification is unique.
\end{prop}

\begin{proof}
  By negating the two sides of the equivalence, the goal becomes proving that a
  proof structure $\pi$ contains a switching cycle if and only if its graphification
  $(G,M)$ contains an alternating cycle.

  Consider any alternating cycle for $M$ in $G$ of length $2n$, and take the
  $\mathbb{Z}/(n)$-indexed sequence of vertices corresponding to the matching
  edges in the cycle. By construction of the graphification, if two edges in $M$
  are incident to a common non-matching edge, then the corresponding links in
  $\pi$ are adjacent: thus, in our sequence, each vertex is adjacent to the
  previous and the next one, and thus we have a cycle. If it were not a
  switching cycle, it would contain three consecutive links $p, q, r$ with $q$ a
  $\bindnasrepma$-link and $p, r$ its predecessors\footnote{To expand on this
    point: this is because we have prohibited vertex repetitions in our
    definition of cycles. This is legitimate since a graph is a forest if and
    only if it does not contain a non-vertex-repeating cycle.}; but then the
  alternating cycle would have to cross two incident non-matching edges (from
  $p$ to $q$ and from $q$ to $r$), which is impossible. Thus, $\pi$ contains a
  switching cycle.

  To show the converse we will exhibit a right inverse to the map from
  alternating cycles to switching cycles defined above. Consider a switching
  cycle: it can be partitioned into directed paths from $\mathtt{ax}$-links to
  $\otimes$-links. Let $l$ be an intermediate link in such a path, and $e, p, s$
  be matching edges corresponding respectively to $l$, its predecessor, and its
  successor in the directed path. $s$ has a unique endpoint $u$ which is
  incident to both endpoints of $e$; $e$ has a unique endpoint $v$ which is
  \emph{not} incident to both endpoints of $p$. To join $e$ with $s$, we use the
  edge $(u,v)$. By taking all these non-matching edges for all maximal directed
  paths in the cycle, as well as a choice of two edges incident to each matching
  edge corresponding to an $\mathtt{ax}$-link, and the matching edges $(a_l,
  b_l)$ corresponding to all the links $l$ in the cycle, we get an alternating
  cycle.
\end{proof}

The situation was a bit nicer for RB-graphs, with an actual bijection between
cycles (Proposition~\ref{prop-rb-bijection}) unlike the case of graphifications.
That said, the main technical advantages of the latter that we sought are
summarized by the following properties.

\begin{lem}%
  \label{lemma-graphification-conclusion}
  Let $\pi$ be a proof structure with graphification $(G,M)$ and $l$ be a link of $\pi$
  such that $(a_l, b_l) \in M$ is a bridge of $G$. Then $l$ is a terminal link
  in $\pi$, and if $l$ is a $\otimes$-link, then removing $l$ from $\pi$
  disconnects its predecessors.
\end{lem}
\begin{proof}
  Suppose for contradiction that $l$ is not a terminal link, and let $l'$ be a
  successor of $l$. Then for some endpoint $v$ of $(a_{l'}, b_{l'})$, $(a_l, v)$
  and $(b_l, v)$ are both edges in $G$, and they make up a path between $a_l$
  and $b_l$ not going through $(a_l, b_l)$. Thus, $(a_l, b_l)$ cannot be a
  bridge.

  The fact that $(a_l,b_l)$ is a bridge means that by removing this edge, $a_l$
  and $b_l$ are in different connected components; if $l$ is a $\otimes$-link,
  each of these connected components contain the matching edge corresponding to
  one predecessor of $l$.
\end{proof}

\begin{thm}%
  \label{graphification-seq-bijection}
  Let $\pi$ be a proof structure and $(G,M)$ be its graphification. There is a
  bijection between the sequent calculus proofs corresponding to $\pi$ (if any)
  and the sequentializations (i.e., the derivation trees for the inductive
  definition of Theorem~\ref{upm-sequentialization}) of $(G,M)$ (if any),
  through which occurrences of Mix rules correspond to disjoint unions and
  conversely.
\end{thm}
This entails, in particular, the analogous property to
Proposition~\ref{proofification-bridge} for graphifications.
\begin{proof}
  We convert a sequentialization $S$ of $(G,M)$ into a sequentialization
  $\Sigma$ of $\pi$ inductively as follows. Since $G \neq \emptyset$, the last
  rule of $S$ is either a disjoint union or the introduction of a bridge $e =
  (a_l,b_l) \in M$ by joining together $(G_a,M_a)$ and $(G_b, M_b)$ with
  respective sequentializations $S_a$ and $S_b$. In the latter case, $l$ is a
  terminal link of $\pi$.
  \begin{itemize}
  \item If $G_a = G_b = \emptyset$, then $l$ is an $\mathtt{ax}$-link, and
    $\Sigma$ consists of a single $\mathtt{ax}$-rule.
  \item If $G_a \neq \emptyset$ and $G_b = \emptyset$, then $l$ is a
    $\bindnasrepma$-link, and the removal of $l$ from $\pi$ yields a proof
    structure $\pi'$ whose graphification is $(G_a,M_a)$. $\Sigma$ then consists of a
    $\bindnasrepma$-rule introducing $l$ applied to the sequentialization of
    $\pi'$ corresponding to $S_a$.
  \item If $G_a \neq \emptyset$ and $G_b \neq \emptyset$, then $l$ is a
    $\otimes$-link. Since $e$ is a bridge, the removal of $l$ from $\pi$ yields
    two proof structures $\pi_a$ and $\pi_b$ whose respective graphifications are $(G_a,
    M_a)$ and $(G_b,M_b)$. $\Sigma$ then consists of an $\otimes$-rule applied
    to the translations of $S_a$ and $S_b$.
  \end{itemize}
  If the last rule of $S$ is a disjoint union rule, it is translated into a Mix
  rule in $\Sigma$.

  The bijectivity can be proven by defining the inverse transformation and by
  checking that it is indeed its inverse.
\end{proof}

In particular, $\pi$ is a MLL+Mix proof net if and only $(G,M)$ admits a
sequentialization, that is, according to Theorem~\ref{upm-sequentialization}, if
and only if $M$ is the only perfect matching of~$G$.
Proposition~\ref{graphification-correctness} tells us that this is equivalent to $\pi$
satisfying the Danos--Regnier acyclicity criterion. Therefore, this criterion 
characterizes MLL+Mix proof nets: as we wanted, we just proved the
sequentialization theorem for MLL+Mix (Theorem~\ref{danos-regnier}).

\subsection{A sequentialization algorithm for MLL+Mix proof nets}%
\label{sec:sequentialization}

In~\Cref{sec:linear-time}, we saw how to decide MLL+Mix correctness in linear
time, matching the known time complexity for MLL correctness. But the algorithms
for MLL correctness still have an advantage: they can compute a
sequentialization in linear time, whereas we only have a decision procedure for
\textsc{MixCorr} which returns a yes/no answer\footnote{It can find a switching
  cycle, witnessing incorrectness, but cannot produce a certificate of
  correctness.}. We do not know how to compute MLL+Mix sequentializations in
linear time. Nevertheless, by applying our bridge between proof nets and graph
theory, we get the first \emph{quasi-linear} time algorithm for \textsc{MixSeq}.
The beginning of the next section will discuss why the problem seems harder with
Mix.

Our algorithm proceeds by first determining the root of the derivation tree and
the link it introduces. To obtain the children of the root, it suffices to
recurse on the connected components created by removing this link.

Furthermore, through the correspondence of Theorem~\ref{graphification-seq-bijection},
finding a link which is introduced by the last rule of some sequentialization
amounts to finding a bridge in the matching of the graphification of the proof net (\cf\
\Cref{sec:graphification}). This is in fact a bit more convenient with graphifications than with
general unique perfect matchings, thanks to the following property:

\begin{lem}%
  \label{sequentialization-lemma}
  All bridges in the graphification of some proof structure are matching edges.
\end{lem}
\begin{proof}
  Let $e$ be a non-matching edge. Then there are matching edges $(u,v)$ and
  $(s,t)$ such that the link corresponding to $(u,v)$ is the predecessor of the
  one for $(s,t)$, and $e = (u,s)$. The non-matching edge $(v,s)$ is then also
  present in the graph, and so $e$ cannot be a bridge.
\end{proof}

The algorithm will alternate between finding and deleting bridges; a deletion
may cut cycles and thus create new bridges, which we want to detect without
traversing the entire graph each time. To do so, we use a \emph{dynamic
  bridge-finding data structure} designed for this kind of use case by Holm et
al.~\cite{holm_dynamic_2018}. It keeps an internal state corresponding to a
graph, whose set of $n$ vertices is immutable but whose set of edges may vary,
and supports the following operations in $O({(\log n)}^2 {(\log \log n)}^2)$
amortized time:
\begin{itemize}
\item updating the graph by inserting or deleting an edge;
\item computing the number of vertices of the connected component of a given vertex;
\item finding a bridge in the connected component of a given vertex;
\item determining whether two vertices are in the same connected component.
\end{itemize}

\begin{thm}%
  \label{quasi-linear-decomposition}
  \textsc{MixSeq} can be solved in $O(n {(\log n)}^2 {(\log \log n)}^2)$ time.
\end{thm}
\begin{proof}
  Let $\pi$ be a MLL+Mix proof net with $n$ links, and $(G = (V,E), M)$ be its
  graphification. Both $V$ and $E$ have cardinality $O(n)$ (in fact, $|V| = 2n$ and
  $|M| = n$).

  The algorithm starts by initializing the bridge-finding data structure $D$
  with the graph~$G$, computing the weakly connected components of $\pi$ in
  linear time, and selecting a link in each component. On each selected link
  $l$, we call the following recursive procedure; its role is to sequentialize
  the sub-proof net of $\pi$ containing $l$ whose graphification is a current connected
  component of $G$ ($G$ and $D$ being mutable global variables):
  \begin{itemize}
  \item Let $u$ be one endpoint of the matching edge corresponding to $l$. Using
    the bridge-finding structure, find a bridge $e = (v,w)$ in the component of
    $u$; necessarily, $e \in M$. Remove the edge $e$ from $G$ (and reflect this
    change on $D$ with a deletion operation).
  \item If both $v$ and $w$ are isolated vertices, $e$ corresponds to an
    $\mathtt{ax}$-link and the entire sub-proof net consisted of this link. In
    this case, return a sequentialization with a single $\mathtt{ax}$-rule.
  \item If one of $v$ and $w$ is isolated, and the other is not --- by symmetry,
    let us assume the latter is $v$ --- then $e$ corresponds to a
    $\bindnasrepma$-link $l'$. Let $p$ and $p'$ be its predecessors.
    \begin{itemize}
    \item Remove all edges incident to $v$.
    \item If the matching edges corresponding to $p$ and $p'$ are in the same
      connected component of $G$, recurse on $p$, add a final
      $\bindnasrepma$-link and return the resulting sequentialization.
    \item If $p$ and $p'$ are in different connected components of $G$, recurse
      on $p$ and $p'$, use the results as the two premises of a Mix rule, add a
      final $\bindnasrepma$-link and return the resulting sequentialization.
    \end{itemize}
  \item If neither $v$ nor $w$ is isolated, $e$ corresponds to a $\otimes$-link.
    This is handled similarly to the $\bindnasrepma$+Mix case above.
  \end{itemize}
  Let us evaluate the time complexity. At each recursive call, one bridge is
  eliminated from $G$, so the number of recursive calls is $n$. The cost of each
  recursive call is $O(1)$ except for the updates and queries of the
  bridge-finding data structure. In total, there are $|E| = O(n)$ deletions,
  $|M| = n$ bridge queries, and at most $n$ connectedness tests, and each of
  those takes $O({(\log n)}^2 {(\log \log n)}^2)$ amortized time. Hence the $O(n
  {(\log n)}^2 {(\log \log n)}^2)$ bound.
\end{proof}

\begin{rem}
  If we want to compute a sequentialization for a unique perfect matching, in
  general, a complication is the existence of bridges which are not in the
  matching.

  Interestingly, one can determine whether a bridge $e$ is in $M$ \emph{without
    looking at $M$}: it is the case if and only if both of the connected
  components created by removing $e$ have an odd number of vertices. This leads
  to an algorithm for \textsc{UniquePM}; it is virtually the same as the one
  proposed by Gabow et al.~\cite[\S2]{gabow_unique_2001}\footnote{Not to be
    confused with their algorithm for
    \textsc{UniquenessPM}~\cite[\S3]{gabow_unique_2001} that we used in
    \Cref{sec:linear-time}. They only claim a bound of $O(m \log^4 n)$ because
    the best dynamic 2-edge-connectivity data structure known at the time has
    operations in $O(\log^4 n)$ amortized time.}, from which we took our
  inspiration.
\end{rem}

\begin{rem}
  One needs to use a sparse representation for derivation trees: the size of a
  fully written-out sequent calculus proof is, in general, not linear in the
  size of its proof net.
\end{rem}

\section{On the kingdom ordering of links}%
\label{sec:kingdom-ordering}

One may wonder if we could not have just tweaked an algorithm for MLL
sequentialization into an algorithm for \textsc{MixSeq}. In order to argue to
the contrary, let us briefly mention a difference between Bellin and van de
Wiele's study of the sub-proof nets of MLL proof nets~\cite{bellin_subnets_1995}
and its extension to the MLL+Mix case by Bellin~\cite{bellin_subnets_1997}. Any
MLL sub-proof net of a MLL proof net may appear in the sequentialization of the
latter; however, for MLL+Mix, Figure~\ref{fig:counterexample} serves as a
counterexample: the sub-proof structure containing all links but the
$\otimes$-link is correct for MLL+Mix, but it cannot be an intermediate step in
a sequentialization of the entire proof net. A \emph{normality} condition is
needed to distinguish those sub-proof nets which may appear in a
sequentialization, and this is why sequentialization algorithms which are
morally based on a greedy parsing strategy, such as Guerrini's linear-time
algorithm~\cite{guerrini_linear_2011}, do not adapt well to the presence of the
Mix rule.

Any link $l$ in a MLL+Mix proof net $\pi$ admits a minimum normal sub-proof net
of $\pi$ containing $l$, its \emph{kingdom}~\cite{bellin_subnets_1997}. Bellin's
\emph{kingdom ordering} is the partial order on links corresponding to the
inclusion between kingdoms. We give an algorithm to compute this order for any
MLL+Mix proof net: this is yet another application of matching theory. It uses a
characterization of the kingdom ordering in terms of a relation called
\emph{dependency} by Bagnol et al.~\cite{bagnol_dependencies_2015} (who, in
turn, take this name from the closely related \emph{dependency graph} of Mogbil
and Naurois~\cite{jacobe_de_naurois_correctness_2011}). We will also see how
this dependency relation can be reformulated, through our correspondence between
proof structures and perfect matchings, in terms of the \emph{blossoms}
mentioned in \Cref{sec:pm} and \Cref{sec:linear-time}.

One may in fact define the kingdom ordering, written $\ll_\pi$, without
reference to the notion of normal sub-proof net (we will not introduce the
latter formally here):

\begin{defi}
  Let $\pi$ be a MLL+Mix proof net. For any two links $p,q$ of $\pi$, $p \ll_\pi
  q$ if and only if, in any sequentialization of $\pi$, the rule introducing $q$
  has, among its premises, a proof net containing $p$.
\end{defi}

From this point of view, the kingdom ordering gives us information about the set
of all sequentializations. Let us give some examples. The proof net of
Figure~\ref{fig:proofification-ex2} admits a unique sequentialization, so this
directly gives us the kingdom ordering: for instance the middle $\otimes$-link
is the greatest element. On the other hand, in the proof net of
Figure~\ref{fig:counterexample}, both $\bindnasrepma$-links may be introduced by a
last rule, so there is no greatest element. In fact, the kingdom ordering
coincides with the predecessor relation. So it does not distinguish between the
3 terminal links even though, unlike the 2 others, the $\otimes$-link cannot be
introduced last.

Before proceeding further, here is another property of MLL proof nets which is
contradicted by Figure~\ref{fig:counterexample} for MLL+Mix proof nets, providing more
evidence that \textsc{MixSeq} is trickier algorithmically than MLL
sequentialization.

\begin{prop}%
  \label{last-rule}
  Let $\pi$ be a MLL proof net and $l$ be a maximal link for $\ll_\pi$. Then
  there exists a sequentialization of $\pi$ whose last rule introduces $l$.
\end{prop}
\begin{proof}
  If $l$ is a terminal $\bindnasrepma$-link, no other assumption is needed for
  the existence of such a sequentialization. Else, $l$ is a terminal
  $\otimes$-link and it suffices to show that $l$ is \emph{splitting}, i.e.,
  that the removal of $l$ splits $\pi$ into two connected components.

  Suppose that it is not the case, and consider some sequentialization of $\pi$:
  it must contain a $\bindnasrepma$-rule, applied to a sub-proof net $\pi'$ for
  which $l$ is splitting, which turns it into a sub-proof net for which $l$ is
  not splitting anymore. Let $p$ be the $\bindnasrepma$-link introduced by that
  rule; its predecessors lie in different connected components of $\pi'
  \setminus \{l\}$. Since $\pi'$ is a MLL proof net, the predecessors of $p$ are
  connected by a switching path in $\pi'$, which must cross $l$. This shows that
  $l$ is a dependency of $p$ in the sense of Definition~\ref{def-dependency},
  contradicting the maximality of $l$. (This only uses the fact that $D(\pi)
  \subseteq \ll_\pi$, which is the \enquote{easy} part of Bellin's theorem.)
\end{proof}

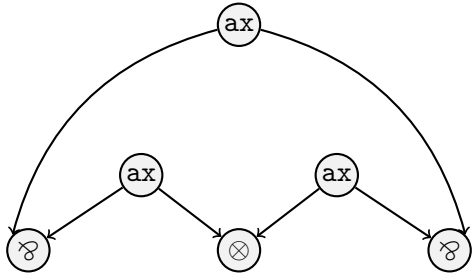
\begin{figure}
  \centering
  \begin{tikzpicture}
    \node[bigvertex] (P1) at (0.5,0) {$\bindnasrepma$};
    \node[bigvertex] (T) at (3.3,0) {$\otimes$};
    \node[bigvertex] (P2) at (6.1,0) {$\bindnasrepma$};
    \node[bigvertex] (Ax1) at (2,1) {$\mathtt{ax}$};
    \node[bigvertex] (Ax2) at (3.3,3) {$\mathtt{ax}$};
    \node[bigvertex] (Ax3) at (4.6,1) {$\mathtt{ax}$};

    \draw[black, thick, ->] (Ax2) to [bend right] (P1.north west);
    \draw[black, thick, ->] (Ax2) to [bend left] (P2.north east);
    \draw[black, thick, ->] (Ax1) -- (P1);
    \draw[black, thick, ->] (Ax1) -- (T);
    \draw[black, thick, ->] (Ax3) -- (T);
    \draw[black, thick, ->] (Ax3) -- (P2);
  \end{tikzpicture}
  \caption{A MLL+Mix proof net which highlights a difficulty in solving
    \textsc{MixSeq}.}%
  \label{fig:counterexample}
\end{figure}

\subsection{Computing the kingdom ordering}%
\label{sec:computing-kingdom}

\begin{defi}%
  \label{def-dependency}
  Let $\pi$ be a proof structure. We write $D(\pi)$ for the \emph{dependency
    relation} defined as follows: for any two links $p \neq q$ of $\pi$,
  \emph{$p$ is a dependency of $q$} when $q$ is a $\bindnasrepma$-link and there
  exists a switching path between the predecessors of $q$ going through $p$.
\end{defi}

For instance, in the proof net of Figure~\ref{fig:proofification-ex2}
(\Cref{sec:proofification}), the left $\bindnasrepma$-link depends on the left
$\otimes$-link, but not on the other $\otimes$-links or $\bindnasrepma$-links;
the middle $\otimes$-link has no dependency. In the case of
Figure~\ref{fig:counterexample}, the dependency relation is empty.

\begin{thm}[Bellin~{\cite[Lemma~2]{bellin_subnets_1997}}\footnote{This theorem
    was rediscovered by Bagnol et
    al.~\cite[Theorem~11]{bagnol_dependencies_2015} in the special case of MLL
    proof nets without Mix (they refer to the kingdom ordering as the
    \enquote{order of introduction}). We borrow the notations $D(\pi)$ and
    $S(\pi)$ from them.}]%
  \label{kingdom-ordering-from-dep}
  Let $\pi$ be a \emph{MLL+Mix proof net}. The transitive closure of $D(\pi)
  \cup S(\pi)$ is $\ll_\pi$, where $(p,q) \in S(\pi)$ means that $p$ is a
  predecessor of $q$.
\end{thm}

The dependency relation can be computed by reduction to a matching problem
\emph{in the case of MLL+Mix proof nets}: even though it is well-defined in
arbitrary proof structures, we need MLL+Mix correctness to compute it, because
our matching algorithm relies on the absence of alternating cycles. It is mostly
a matter of applying a lemma from our paper~{\cite{nguyen_constrained_2019}};
since the latter has not been peer-reviewed as of the time of writing, we
reproduce the proof in the appendix.

\begin{lem}[{\cite{nguyen_constrained_2019}} / \Cref{appendix-proof-prescribed}]%
  \label{lemma-prescribed}
  Let $M$ be a matching of some graph $G = (V,E)$. Suppose that:
  \begin{itemize}
  \item there are \emph{no alternating cycles} for $M$ --- equivalently, $M$ is
    the unique perfect matching of the subgraph induced by the vertices matched
    by $M$;
  \item there are exactly two unmatched vertices $u,v$.
  \end{itemize}
  Then the existence of an alternating path for $M$ with endpoints $u,v$ and
  \emph{crossing a prescribed matching edge $e \in M$} can be reduced in
  $\mathsf{AC^0}$ to the existence of a perfect matching; furthermore, such a
  path can be found in linear time.
\end{lem}

\begin{rem}
  An alternating path between unmatched vertices is often called an
  \emph{augmenting path}; combinatorial maximum matching algorithms generally
  work by iteratively searching for augmenting paths, see, \eg,~\cite[Chapter~9]{tarjan_data_1983}.
\end{rem}

\begin{thm}%
  \label{dependency-linear-time}
  Let $\pi$ be a \emph{MLL+Mix proof net} with a link $p$ and a
  $\bindnasrepma$-link $q$. Deciding whether $(p,q) \in D(\pi)$ can be done in
  linear time, in randomized $\mathsf{NC}$ and in
  $\mathsf{quasiNC}$.
\end{thm}

\begin{proof}
  A degenerate case is when $p$ is a predecessor of $q$: in this case, $p$
  depending on $q$ is equivalent to $\pi$ becoming incorrect if $q$ is turned
  into a $\otimes$-link, and thus the complexity is the same as that of (the
  complement of) the correctness problem.

  When $p$ is not a predecessor of $q$, the definition of dependency translates
  into the problem defined in the above lemma by taking the graphification of
  $\pi$, and removing the matching edge corresponding to $q$. The endpoints of
  this edge then become unmatched, and we choose as prescribed intermediate edge
  the matching edge corresponding to $p$. The fact that $\pi$ is a proof net
  ensures that the acyclicity assumption of Lemma~\ref{lemma-prescribed} is
  satisfied.

  We directly obtain the linear time complexity, and since the existence of a
  perfect matching can be decided in randomized $\mathsf{NC}$ or
  $\mathsf{quasiNC}$ (\cf\ \Cref{sec:pm}), so can our problem.
\end{proof}

A transitive closure can be computed in polynomial time, and reachability in a
directed graph can be decided in $\mathsf{NL} \subset \mathsf{quasiNC}$, so we
get in the end:

\begin{cor}%
  \label{order-intro-polytime}
  There are a polynomial-time algorithm and a $\mathsf{quasiNC}$
  algorithm to compute the kingdom ordering $\ll_\pi$ of any MLL+Mix proof net
  $\pi$.
\end{cor}

\subsection{Dependencies and blossoms in unique perfect matchings}%
\label{sec:blossoms}

We will now see how, through the correspondence of \Cref{sec:equivalence},
Bellin's theorem can be rephrased as a statement on unique perfect matchings.

\begin{defi}
  Let $G$ be a graph and $M$ be a perfect matching of $G$. A \emph{blossom} for
  $M$ is a cycle whose vertices are all matched within the cycle, except for
  one, its \emph{root}. The matching edge incident to the root is called the
  \emph{stem} of the blossom.
\end{defi}

\begin{figure}
  \centering
      \begin{tikzpicture}
      \node[vertex] (s) at (0,0) {};
      \node[vertex] (a) at (2.5,0) {};
      \node[vertex] (b) at (4,-2) {};
      \node[vertex] (c) at (6,-1) {};
      \node[vertex] (d) at (6,1) {};
      \node[vertex] (e) at (4,2) {};

      \draw[matching edge] (s) -- (a);
      \draw[non matching edge] (a) -- (b);
      \draw[matching edge] (b) -- (c);
      \draw[non matching edge] (c) -- (d);
      \draw[matching edge] (d) -- (e);
      \draw[non matching edge] (e) -- (a);
    \end{tikzpicture}
  \caption{A blossom of length 5, with its stem on the left.}%
  \label{fig:blossom}
\end{figure}
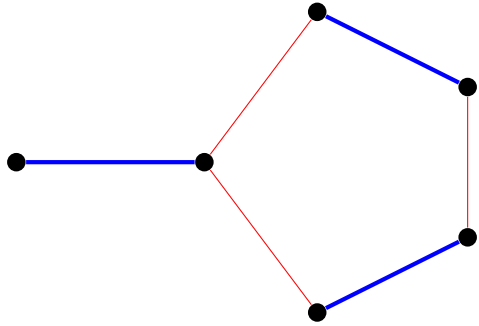

That is, a blossom consists of an alternating path between two vertices,
starting and ending with a matching edge, together with a non-matching edge from
the root to each of these two vertices. See Figure~\ref{fig:blossom} for an
illustration; as another example, in Figure~\ref{fig:pm-unique}, the two
triangles are blossoms with a common stem. The stem of a blossom is not part of
the cycle. Blossoms are central to combinatorial matching algorithms, \eg,~\cite{edmonds_paths_1965, gabow_unique_2001}, as we have previously
mentioned.

\begin{defi}
  When $e \in M$ is in some blossom with stem $f \in M$, we write $e \rightarrow
  f$.
\end{defi}

This is the graph-theoretical counterpart of the dependency relation, as is
shown by the following two propositions.

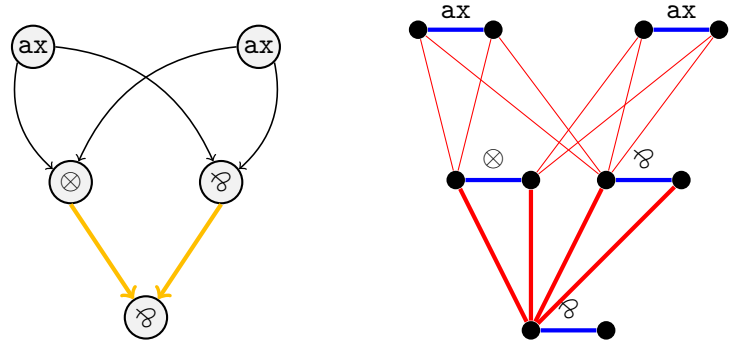
\begin{figure}
  \centering
      \begin{tikzpicture}
      \node[bigvertex] (Ax1) at (1.5,3.6) {$\mathtt{ax}$};
      \node[bigvertex] (Ax2) at (4.5,3.6) {$\mathtt{ax}$};

      \node[bigvertex] (T) at (2,1.8) {$\otimes$};
        \draw[black, semithick, ->] (Ax1.south west) to [bend right] (T);
        \draw[black, semithick, ->] (Ax2.west) to [bend right] (T);

      \node[bigvertex] (P1) at (4,1.8) {$\bindnasrepma$};
        \draw[black, semithick, ->] (Ax1.east) to [bend left] (P1);
        \draw[black, semithick, ->] (Ax2.south east) to [bend left] (P1);

        \node[bigvertex] (P2) at (3,0) {$\bindnasrepma$};
        \draw[amber, ultra thick, ->] (P1.south) -- (P2);
        \draw[amber, ultra thick, ->] (T.south) -- (P2);
  \end{tikzpicture}
  \qquad\qquad
    \begin{tikzpicture}
        \node[vertex] (Ax1l) at (1,4) {};
        \node[vertex] (Ax1r) at (2,4) {};
        \draw[matching edge] (Ax1l) -- node[above, text=black] {$\mathtt{ax}$} ++ (Ax1r);
        \node[vertex] (Ax2l) at (4,4) {};
        \node[vertex] (Ax2r) at (5,4) {};
        \draw[matching edge] (Ax2l) -- node[above, text=black] {$\mathtt{ax}$} ++ (Ax2r);

          \node[vertex] (Tl) at (1.5,2) {};
          \node[vertex] (Tr) at (2.5,2) {};
          \draw[matching edge] (Tl) -- node[above, text=black] {$\otimes$} ++ (Tr);
          \draw[non matching edge] (Ax1l) -- (Tl);
          \draw[non matching edge] (Ax1r) -- (Tl);
          \draw[non matching edge] (Ax2l) -- (Tr);
          \draw[non matching edge] (Ax2r) -- (Tr);

          \node[vertex] (P1l) at (3.5,2) {};
          \node[vertex] (P1r) at (4.5,2) {};
          \draw[matching edge] (P1l) -- node[above, text=black] {$\bindnasrepma$} ++ (P1r);
          \draw[non matching edge] (Ax1l) -- (P1l);
          \draw[non matching edge] (Ax1r) -- (P1l);
          \draw[non matching edge] (Ax2l) -- (P1l);
          \draw[non matching edge] (Ax2r) -- (P1l);

          \node[vertex] (P2l) at (2.5,0) {};
          \node[vertex] (P2r) at (3.5,0) {};
          \draw[matching edge] (P2l) -- node[above, text=black] {$\bindnasrepma$} ++ (P2r);

            \draw[non matching edge, ultra thick] (Tl) -- (P2l);
            \draw[non matching edge, ultra thick] (Tr) -- (P2l);
            \draw[non matching edge, ultra thick] (P1l) -- (P2l);
            \draw[non matching edge, ultra thick] (P1r) -- (P2l);
    \end{tikzpicture}
    \caption{The proof net of Figure~\ref{fig:proof-structure} and its graphification
      (\cf~Figure~\ref{fig:graphification-example}); the directed edges of the proof
      net correspond to blossoms of length 3 in its graphification.}%
  \label{fig:blossom-subformula}
\end{figure}

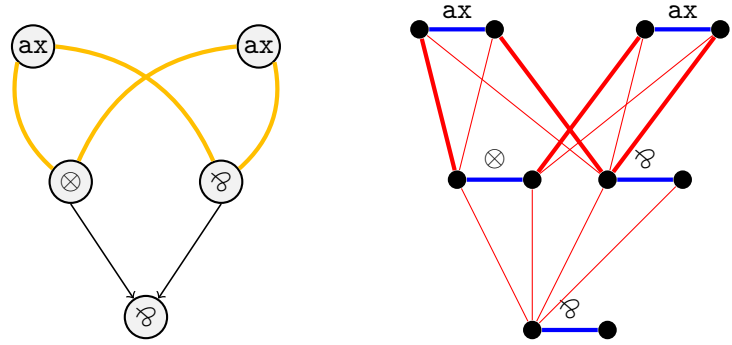
\begin{figure}
  \centering
      \begin{tikzpicture}
      \node[bigvertex] (Ax1) at (1.5,3.6) {$\mathtt{ax}$};
      \node[bigvertex] (Ax2) at (4.5,3.6) {$\mathtt{ax}$};
      \node[bigvertex] (T) at (2,1.8) {$\otimes$};
      \node[bigvertex] (P1) at (4,1.8) {$\bindnasrepma$};
      \node[bigvertex] (P2) at (3,0) {$\bindnasrepma$};

        \draw[black, semithick, ->] (P1.south) -- (P2);
        \draw[black, semithick, ->] (T.south) -- (P2);
        \draw[amber, ultra thick] (Ax1.south west) to [bend right] (T);
        \draw[amber, ultra thick] (Ax2.west) to [bend right] (T);
        \draw[amber, ultra thick] (Ax1.east) to [bend left] (P1);
        \draw[amber, ultra thick] (Ax2.south east) to [bend left] (P1);

  \end{tikzpicture}
  \qquad\qquad
    \begin{tikzpicture}
        \node[vertex] (Ax1l) at (1,4) {};
        \node[vertex] (Ax1r) at (2,4) {};
        \draw[matching edge] (Ax1l) -- node[above, text=black] {$\mathtt{ax}$} ++ (Ax1r);
        \node[vertex] (Ax2l) at (4,4) {};
        \node[vertex] (Ax2r) at (5,4) {};
        \draw[matching edge] (Ax2l) -- node[above, text=black] {$\mathtt{ax}$} ++ (Ax2r);

          \node[vertex] (Tl) at (1.5,2) {};
          \node[vertex] (Tr) at (2.5,2) {};
          \draw[matching edge] (Tl) -- node[above, text=black] {$\otimes$} ++ (Tr);
          \draw[non matching edge] (Ax1l) -- (Tl);
          \draw[non matching edge] (Ax1r) -- (Tl);
          \draw[non matching edge] (Ax2l) -- (Tr);
          \draw[non matching edge] (Ax2r) -- (Tr);

          \node[vertex] (P1l) at (3.5,2) {};
          \node[vertex] (P1r) at (4.5,2) {};
          \draw[matching edge] (P1l) -- node[above, text=black] {$\bindnasrepma$} ++ (P1r);
          \draw[non matching edge] (Ax1l) -- (P1l);
          \draw[non matching edge] (Ax1r) -- (P1l);
          \draw[non matching edge] (Ax2l) -- (P1l);
          \draw[non matching edge] (Ax2r) -- (P1l);

          \node[vertex] (P2l) at (2.5,0) {};
          \node[vertex] (P2r) at (3.5,0) {};
          \draw[matching edge] (P2l) -- node[above, text=black] {$\bindnasrepma$} ++ (P2r);
          \draw[non matching edge] (Tl) -- (P2l);
          \draw[non matching edge] (Tr) -- (P2l);
          \draw[non matching edge] (P1l) -- (P2l);
          \draw[non matching edge] (P1r) -- (P2l);

            \draw[non matching edge, ultra thick] (Ax1l) -- (Tl);
            \draw[non matching edge, ultra thick] (Ax2l) -- (Tr);
            \draw[non matching edge, ultra thick] (Ax2r) -- (P1l);
            \draw[non matching edge, ultra thick] (Ax1r) -- (P1l);
    \end{tikzpicture}

    \caption{A blossom of length 7 corresponding to a dependency. The yellow
      cycle is not a switching cycle, but should be seen as a switching path
      between both predecessors of the $\bindnasrepma$-link.}%
  \label{fig:blossom-dependency}
\end{figure}

\begin{prop}%
  \label{blossom-graphification}
  Let $\pi$ be a MLL+Mix proof net and $(G,M)$ be its graphification. Let $p, q$
  be links in~$\pi$ with corresponding matching edges $e_p, e_q \in M$. Then
  $e_p \rightarrow e_q$ if and only if $p$ is a dependency of $q$ or a
  predecessor of $q$, i.e., $(p, q) \in D(\pi) \cup S(\pi)$.
\end{prop}

Both the cases $(p,q) \in S(\pi)$ and $(p,q) \in D(\pi)$ occur in the proof net
of Figure~\ref{fig:proof-structure}, see respectively Figure~\ref{fig:blossom-subformula}
and Figure~\ref{fig:blossom-dependency}.

\begin{proof}
  If $(p,q) \in S(\pi)$, then by construction there exists a blossom of length 3
  containing $p$ with stem $q$. If $(p,q) \in D(\pi)$, then for the same reason
  as Proposition~\ref{graphification-correctness}, we can get, from the switching path
  between the predecessors of $q$ visiting $p$, an alternating path for $M$
  starting and ending with the edges corresponding to those predecessors and
  crossing the edge corresponding to $p$. By adding two non-matching edges to
  the same endpoint of the matching edge for $q$, we get a blossom with stem
  $q$.

  Conversely, let $q$ be a link, $e$ the corresponding matching edge, and $B$ be
  a blossom with stem $q$. Let us first note that if $B$ contains a
  non-matching edge joining $e$ with the matching edge corresponding to a
  successor of $q$, then by replacing this non-matching edge with its twin
  incident to the other endpoint of $q$, we get an alternating cycle; this is
  impossible because we have assumed $\pi$ to be a MLL+Mix proof net. Therefore,
  the first and last matching edges in $B$ are both precedessors of $q$. If they are
  the same --- that is, if $B$ has length 3 and contains a single matching edge
  --- then this edge corresponds to a predecessor $p$ of $q$. Otherwise, $B$
  gives an alternating path between two distinct predecessors of $q$; necessarily
  $q$ is a $\bindnasrepma$-link (otherwise, there would be an alternating
  cycle), and all links corresponding to matching edges in $B$ are dependencies
  of $q$.
\end{proof}

\begin{prop}%
  \label{blossom-proofification}
  Let $G$ be a graph, $M$ be a perfect matching of $G$ and $\pi$ be the proofification
  of $(G,M)$. Let $e,f \in M$ with corresponding $\otimes$-links $l_e,l_f \in
  M$. Then $e \rightarrow f$ if and only if $l_e$ is a dependency of some
  $\bindnasrepma$-link $q$ from which $l_f$ is reachable (by a directed path).
\end{prop}
\begin{proof}
  Let $B$ be a blossom with stem $f$, whose two non-matching edges incident to
  $f$ are $a$ and $b$. $B$ translates into a switching path between
  $\mathtt{ax}_a$ and $\mathtt{ax}_b$ in $\pi$. Now, $\mathtt{ax}_a$ and
  $\mathtt{ax}_b$ are also leaves of a binary tree of $\bindnasrepma$-links
  whose root has the single successor $l_f$; by taking $q$ to be the lowest
  common ancestor of $\mathtt{ax}_a$ and $\mathtt{ax}_b$ in this tree, $l_f$ is
  reachable from $q$, and every link in the path between $\mathtt{ax}_a$ and
  $\mathtt{ax}_b$ depends on $q$. Conversely, any switching path between the two
  predecessors of a $\bindnasrepma$-link corresponds to a blossom for $M$ in $G$.
\end{proof}

\begin{rem}
  In Proposition~\ref{blossom-graphification}, the \enquote{if} direction holds even for
  incorrect proof structures; in Proposition~\ref{blossom-proofification}, note that no
  uniqueness property is required of the perfect matching.
\end{rem}

Thus, we see that Bellin's theorem is equivalent to the following theorem where
$\rightarrow^+$ is the transitive closure of $\rightarrow$.

\begin{thm}
  Let $G$ be a graph with a \emph{unique} perfect matching $M$, and $e, f \in
  M$. The edge $e$ occurs before $f$ in all sequentializations for $M$ if and
  only if $e \rightarrow^+ f$.
\end{thm}

For instance, in Figure~\ref{fig:pm-unique}, the middle edge $e$ is the only
bridge, and it is the stem of the two triangular blossoms which contain the
other matching edges.

This graph-theoretic version is somewhat simpler to state than the original
theorem: one takes the transitive closure of a single relation, instead of a
union of two unrelated relations. And as far as we know, this is a new result in
graph theory. We have included it in the companion
paper~\cite{nguyen_constrained_2019}, aimed at a broader audience of graph
theorists, where we present a direct combinatorial proof with no mention of
proof nets.

\section{A reconstruction of RB-graphs via forbidden transitions}%
\label{sec:rb-graph-bis}

In this section, we come back to Retoré's RB-graphs (\Cref{sec:rb-graph}) and
factorize Retoré's correctness criterion (Corollary~\ref{cor-retore}) as a
composition of:
\begin{itemize}
\item the Danos--Regnier correctness graph 
  (Definition~\ref{def-correctness-graph});
\item a reduction to the \textsc{UniquenessPM} problem for a general notion of
  constrained cycles, namely \emph{closed trails avoiding forbidden transitions}.
\end{itemize}
We introduced the latter in~\cite{nguyen_constrained_2019}, but here the logical
order of exposition is the reverse of the order of discovery: it was by
attempting to understand Retoré's RB-graphs that we found this reduction.

\begin{defiC}[{\cite{szeider_finding_2003}}]%
  \label{def:transition}
  Let $G = (V,E)$ be a graph. A \emph{transition graph} for a vertex $v \in V$
  is a graph whose vertices are the edges incident to $v$: $T(v) =
  (\partial(v), E_v)$. A \emph{transition system} on $G$ is a family $T =
  {(T(v))}_{v \in V}$ of transition graphs.

  A graph equipped with a transition system is called a \emph{graph with
    forbidden transitions}.

  A path $v_1, e_1, v_2 \ldots, e_{k-1}, v_k$ is said to be \emph{compatible} if
  for $i = 1, \ldots, k-1$, $e_i$ and $e_{i+1}$ are adjacent in $T(v_{i+1})$.
  For a \emph{cycle}, we also require $e_{k-1}$ and $e_1$ to be adjacent in
  $T(v_1) = T(v_k)$. (So the edges of $T(v)$ actually specify the \emph{allowed}
  transitions.)
\end{defiC}

\begin{rem}
  By \enquote{transition} we mean a pair of consecutive edges in a path/cycle. A
  transition system could equivalently be specified by literally giving the set
  of forbiden transitions, i.e., of edge pairs that cannot occur consecutively.
  This generalizes paired graphs (Definition~\ref{def-paired-graph}) by dropping
  the disjointness requirement on pairs: switching graphs do not make sense
  anymore, but switching cycles (generalized to compatible cycles) still do.
\end{rem}

Finding a compatible path is proved to be $\mathsf{NP}$-complete
in~\cite{szeider_finding_2003}. As in the case of alternating paths in matchings
(cf.\ Remark~\ref{rem-alternating-difficult}), the difficulty is in the
interaction of a local constraint --- any transition (pair of consecutive edges)
must be allowed --- and a global one, namely the fact there must be no repeated
vertices. Indeed, recall from \Cref{sec:terminology} that by definition, paths
and cycles cannot repeat vertices twice. This is important to ensure that
Berge's lemma for alternating cycles (Lemma~\ref{berge}) holds. Following the
terminology of~\cite[\S 1.4]{bang-jensen_digraphs._2009}, let us introduce a
relaxation of this global condition.
\begin{defi}
  A \emph{trail} (resp.\ \emph{closed trail}) is a \enquote{path} (resp.\
  \enquote{cycle}) in which we allow vertex repetitions, but \emph{edge
    repetitions are prohibited}. Compatible trails and compatible closed trails
  in a graph with forbidden transitions are defined analogously to the above
  definition.
\end{defi}
\begin{rem}
  For perfect matchings, an alternating cycle is the same as an alternating
  closed trail: repeating a vertex would imply repeating its unique matching
  edge. However, this is not true for general graphs with forbidden transitions:
  see Figure~\ref{fig:paired-graph-cycle-a} for an example with compatible closed
  trails, but no compatible cycles.
\end{rem}

The relevance of this notion is that for compatible \emph{trails}, we showed the
problem to be tractable~\cite{nguyen_constrained_2019} by using an
\enquote{edge-colored line graph} construction. This construction has other uses
but, in the case of compatible (closed) trails, it can be replaced by a version
using perfect matchings that we define below --- which in fact is the
edge-colored line graph composed with a previously known reduction,
see~\cite{nguyen_constrained_2019} for details. All this arguably goes to show
that the objects which we manipulate are not contrived to fit with RB-graphs:
they arise naturally from other considerations.

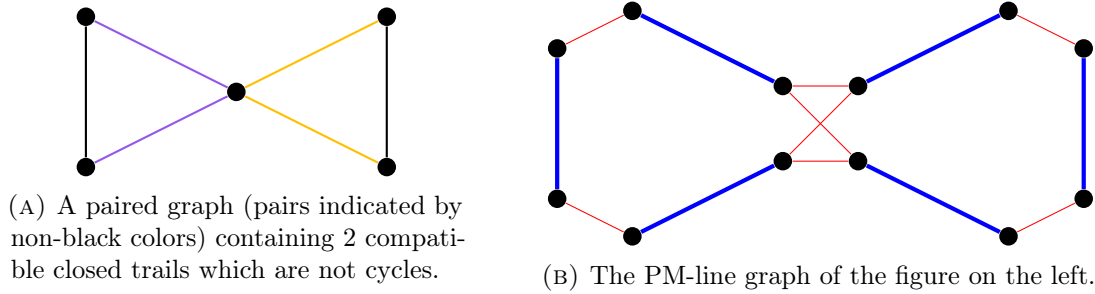
\begin{figure}
  \centering
  \begin{subfigure}{6cm}
    \centering
    \begin{tikzpicture}
      \node[vertex] (w) at (0,2) {};
      \node[vertex] (x) at (4,2) {};
      \node[vertex] (y) at (0,0) {};
      \node[vertex] (z) at (4,0) {};
      \node[vertex] (o) at (2,1) {};
      \draw[thick] (x) -- (z);
      \draw[thick] (w) -- (y);
      \draw[amber,thick] (x) -- (o);
      \draw[amber,thick] (z) -- (o);
      \draw[lavenderindigo,thick] (w) -- (o);
      \draw[lavenderindigo,thick] (y) -- (o);
    \end{tikzpicture}
    \caption{A paired graph (pairs indicated by non-black colors) containing 2
      compatible closed trails which are not cycles.}%
    \label{fig:paired-graph-cycle-a}
  \end{subfigure}\qquad\begin{subfigure}{8cm}
    \centering
    \begin{tikzpicture}
      \node[vertex] (w) at (0,3) {};
      \node[vertex] (w') at (-1,2.5) {};
      \node[vertex] (x) at (5,3) {};
      \node[vertex] (x') at (6,2.5) {};
      \node[vertex] (y) at (0,0) {};
      \node[vertex] (y') at (-1,0.5) {};
      \node[vertex] (z) at (5,0) {};
      \node[vertex] (z') at (6,0.5) {};
      \node[vertex] (ox) at (3,2) {};
      \node[vertex] (oy) at (2,1) {};
      \node[vertex] (oz) at (3,1) {};
      \node[vertex] (ow) at (2,2) {};
      \draw[matching edge] (x') -- (z');
      \draw[matching edge] (w') -- (y');
      \draw[matching edge] (x) -- (ox);
      \draw[matching edge] (z) -- (oz);
      \draw[matching edge] (w) -- (ow);
      \draw[matching edge] (y) -- (oy);
      \draw[non matching edge] (x) -- (x');
      \draw[non matching edge] (y) -- (y');
      \draw[non matching edge] (z) -- (z');
      \draw[non matching edge] (w) -- (w');
      \draw[non matching edge] (ox) -- (oy);
      \draw[non matching edge] (ox) -- (ow);
      \draw[non matching edge] (oz) -- (oy);
      \draw[non matching edge] (oz) -- (ow);
    \end{tikzpicture}
    \caption{The PM-line graph of the figure on the left.}%
    \label{fig:paired-graph-cycle-b}
  \end{subfigure}
  \caption{A graph with forbidden transitions and its PM-line graph.}%
  \label{fig:paired-graph-cycle}
\end{figure}
\begin{defi}%
  \label{def:lpm}
  Let $G$ be a graph and $T$ be a transition system on $G$. The
  \emph{PM-line graph} $L_{PM}(G,T)$ is defined as the graph:
  \begin{itemize}
  \item with vertex set $\{ u_e \mid e \in E,\, u \text{ is an endpoint of } e \}$;
  \item with edge set $M \sqcup E'$, where
    \[ M = \{ (u_e, v_e) \mid e = (u,v) \in E \} \quad E' = \{ (u_e, u_f) \mid
      u \in V,\, e,f \in \partial(u) \text{ are adjacent in } T(u)\} \]
  \item equipped with the perfect matching $M$.
  \end{itemize}
\end{defi}

\begin{propC}[{\cite{nguyen_constrained_2019}}]%
  \label{trail-bijection}
  Closed trails of length $k$ in $G$ compatible with $T$ correspond bijectively
  to alternating cycles of length $2k$ in $L_{PM}(G,T)$.
\end{propC}

An example is given by Figure~\ref{fig:paired-graph-cycle-b}: it contains two
alternating cycles corresponding to the compatible closed trails of
Figure~\ref{fig:paired-graph-cycle-a}.

Finally, we relate the PM-line graph construction to RB-graphs.
\begin{prop}
  Let $\pi$ be a proof structure \emph{with conclusions}
  (Definition~\ref{def-proof-structure-conclusions}) and $C(\pi)$ its
  correctness graph (adapting Definition~\ref{def-correctness-graph} to handle
  conclusion vertices/edges). Let $T$ be the transition system corresponding to
  the paired edges of $C(\pi)$.

  Then $L_{PM}(C(\pi),T)$ is exactly the RB-graph for $\pi$.
\end{prop}
\begin{proof}
  Immediate by comparing Figure~\ref{fig:rb-graph-rules} with Definition~\ref{def:lpm}.
\end{proof}
The moral of the story is that the actual function of RB-graphs is to detect
\emph{compatible closed trails}. It turns out that for the correctness graphs of
proof structures, this is the same as switching cycles, but as
Figure~\ref{fig:paired-graph-cycle-a} shows this is not true in general. The
particularity of correctness graphs that entails this equivalence is that if a
vertex is incident to two edges that are paired together, then it is incident to
at most one unpaired edge (which corresponds to the outgoing edge of a
$\bindnasrepma$ link in the proof structure).

\section{Conclusion}

We have presented a correspondence between proof nets and perfect matchings, and
demonstrated its usefulness through several applications of graph theory to
linear logic: our results give the best known complexity for MLL+Mix correctness
and sequentialization, by taking advantage of sophisticated graph algorithms.
Beyond that, we have also contextualized this correctness problem as a member of
a family of equivalent constrained cycle-finding problems in graphs, and used
this to shed some light on earlier work on proof nets. These connections also
have some benefits for graph theory, as the rephrasing of Bellin's theorem and
our discovery of the \enquote{PM-line graph} construction illustrate; this is
what we attempt to demonstrate in the companion
paper~\cite{nguyen_constrained_2019}. In general, we hope to see fruitful
interactions arise between those two domains.

\subsection{Further hardness results: pomset logic and visible acyclicity}

To take advantage of this connection, one can peruse the literature on graphs to
look for results with potential applications to proof nets. For instance, there
is a $\mathsf{NP}$-hardness result for a certain constrained path-finding
problem on arc-colored \emph{directed graphs}~\cite{gourves_complexity_2013}.
From this, we deduced in~\cite{nguyen_constrained_2019} that finding an
alternating circuit --- for a certain notion of perfect matching in a directed
graph --- is $\mathsf{NP}$-hard. In other words, the absence of alternating
circuits --- one possible generalization of the \textsc{UniquenessPM} problem to
directed graphs (for which Berge's lemma doesn't hold) --- is
$\mathsf{coNP}$-hard.

It turns out that circuits (i.e., directed cycles) also appear in the study of
proof nets:
\begin{itemize}
\item Retoré's \emph{pomset logic}~\cite{goos_pomset_1997} is a conservative
  extension of MLL+Mix with a self-dual non-commutative connective
  $\triangleleft$. The extension of the Danos--Regnier correctness criterion to 
  pomset logic proof nets allows both premises of a $\triangleleft$-link to be
  traversed consecutively by a \enquote{switching cycle}, but \emph{only if the
    left premise is taken before the right one}: the direction of the cycle
  therefore becomes relevant. In~\cite{nguyen_complexity_2020}, we show that
  \emph{the correctness problem for pomset logic is
    $\mathsf{coNP}$-complete}\footnote{This contradicts (assuming that
    $\mathsf{P} \neq \mathsf{NP}$) the polynomial time claim
    of~\cite[Proposition~5]{goos_pomset_1997}, whose purported proof relies on a
    \enquote{standard breadth search algorithm} to find an alternating path for
    a perfect matching in a digraph. Remark~\ref{rem-alternating-difficult}
    explains the subtle issue with this argument.

    A similar mistake appears in Hughes's paper on combinatorial proofs: he
    claims that a \enquote{simple breadth-first
      search}~\cite[Footnote~3]{hughes_proofs_2006} can determine, in linear
    time, some condition that amounts to the correctness of a MLL+Mix proof
    structure. In that case, the mistake is harmless, thanks to our
    Theorem~\ref{thm-linear-time}.}, by adapting our proofification construction
  to take directed graphs as input. A direct proof of the $\mathsf{NP}$-hardness
  of the directed alternating cycle problem is also provided
  in~\cite{nguyen_complexity_2020}.
\item The \emph{visible acyclicity} condition was first introduced by Pagani for
  as a relaxation of the usual correctness criterion for MELL+Mix (by MELL we
  mean Multiplicative-\emph{Exponential} Linear Logic) proof
  structures~\cite{pagani_acyclicity_2006}; it was later extended to
  differential interaction nets~\cite{pagani_visible_2012}. It is defined as the
  absence of certain \enquote{visible cycles}, which become directed when
  exponential boxes are present. One can show that visible acyclicity is
  $\mathsf{coNP}$-hard (a result that we first announced at the DICE 2018
  workshop) by imitation of the proof for pomset logic. However, we do not know
  whether it is in $\mathsf{coNP}$.
\end{itemize}
Interestingly, both pomset logic correctness and visible acyclicity were
motivated by semantic considerations: they are necessary and sufficient
conditions for the soundness of the denotation of proof structures in coherence
spaces.

\subsection{Open questions}

Now that we have shed a new light on MLL+Mix proof nets, it would be interesting
to revisit the well-studied theory of MLL proof nets. Therefore, we would like
to find the right graph-theoretical counterpart to the connectedness condition
in the Danos--Regnier criterion for MLL\@. The goal would be to extract the 
combinatorial essence of the statics of MLL proof structures, \emph{forgetting
  about logic}; without having to handle the dynamics (cut-elimination), one
could hope to distill some simpler combinatorial object, in the same way that
perfect matchings are simpler than MLL+Mix proof structures.

But unique perfect matchings do not seem to be the right setting to do so; and
one year after the conference version of this paper, despite the connections
described here with, \eg, forbidden transitions, we still have not found a
natural graph-theoretic decision problem equivalent to correctness for MLL
without Mix. (As far as naturality is concerned, perfect matchings set a high
bar, given their importance in discrete mathematics!)

Here by \enquote{equivalent} we mean, in particular, through low-complexity
reductions (hopefully computable both in linear time and in $\mathsf{AC}^0$).
Though the $\mathsf{NL}$-completeness of MLL correctness means that it is
equivalent to directed reachability, Mogbil and Naurois's correctness
criterion~\cite{jacobe_de_naurois_correctness_2011} uses a subroutine for
connectivity in undirected forests, a $\mathsf{L}$-complete problem, in its
reduction. A related question is to understand why all known linear-time
correctness criteria for MLL --- including the one presented here --- rely on the
same sophisticated data structure, as mentioned in
Remark~\ref{rem-alternating-difficult}. (Namely, \enquote{incremental tree set
  union}: a restricted union-find data structure with $O(1)$ amortized
operations.)

In the same vein, the present paper does not treat at all --- except for
  the short Remark~\ref{rem-contractibility} --- the \emph{contractibility}
criterion introduced by Danos~\cite{danos_logique_1990}, despite its importance
in recent developments in proof nets (\eg,~\cite{hughes_conflict_2016,bellin_proof_2018}). It is also part of the divide
between MLL and MLL+Mix proof nets: contractibility, reformulated as graph
parsing, underlies a linear-time sequentialization algorithm for
MLL~\cite{guerrini_linear_2011}, while no such algorithm is known for MLL+Mix.
Aside from the obvious question of sequentalizing MLL+Mix nets in linear time,
looking for a mainstream graph-theoretic account of contractibility is also of
interest.

Another question\footnote{This was suggested to the author by Gianluigi Bellin.}
would be to give a graph-theoretic account of the notion of \emph{empire} in
proof nets, similarly to our treatment of kingdoms in
\Cref{sec:kingdom-ordering}. Empires were used in Girard's original proof of the
first correctness criterion (the so-called \enquote{long trip}
criterion)~\cite{girard_linear_1987}; while the kingdom of a link $l$ in a
MLL+Mix proof net is the minimum normal subnet having $l$ as a conclusion, the
empire of $l$ is, dually, the maximum such subnet. To achieve this goal, the
obvious place to start would be the characterization of empires in proof nets
given in~\cite[Lemma~3]{bellin_subnets_1997} using certain paths
(\enquote{chains}) in proof nets.

\subsection{Other variants of proof nets through the lens of graph theory}

We gather here miscellaneous ideas on extending the graph-theoretic viewpoint
beyond MLL+Mix, that we have not had the time to pursue further. Any assertion
that we make below should therefore be seen as purely speculative.

\subsubsection{Jumps and quantifiers}

We have argued that our graphification construction (\Cref{sec:graphification})
faithfully reflects the intrinsic order of logical rules in a proof net. It should
therefore be possible to incorporate \emph{jumps}, which are a way to prescribe
sequentiality constraints on proof nets. By doing so, one would extend our
results to MLL+Mix with (first-order or second-order) quantifiers
$\forall/\exists$: the technology of jumps was first introduced to handle proof
nets with quantifiers~\cite{girard_quantifiers_1991}. This treatment should also
accomodate more general uses of jumps such
as~\cite{di_giamberardino_proof_2008}.

\subsubsection{Essential nets}%
\label{sec:essential-nets}

Larmarche's \emph{essential nets} for \emph{intuitionistic} MLL admit a
correctness criterion formulated using a standard notion on graphs, namely the
\emph{domination} between vertices in a control flow graph. This is at the heart
of Murawski and Ong's linear time algorithm for MLL
correctness~\cite{murawski_fast_2006}. So it would be interesting, in view of
the aforementioned goal of understanding why the \enquote{incremental tree set
  union} data structure of~\cite{gabow_linear-time_1985} seems necessary to
decide MLL correctness in linear time (it occurs in the computation of a
\enquote{dominator tree} in~\cite{murawski_fast_2006}), to compare this
domination criterion with the criteria based on unique perfect matchings.

A first remark is that, via a classical correspondence between directed graphs
and graphs equipped with \emph{bipartite} perfect matchings, the essential net
obtained from a MLL proof structure by the reduction
of~\cite{murawski_fast_2006} (the so-called \enquote{trip translation}) can be
identified with a maximal bipartite subgraph of its RB-graph. The missing piece
is to understand whether this is an instance of a purely graph-theoretic
reduction from the domination condition to the \textsc{UniquenessPM} problem.

\section*{Acknowledgments}
\noindent This work started as a side project during an internship in the
Operations Research team at the Laboratoire d'Informatique de Paris 6,
supervised by Christoph Dürr, who taught the author the expressive power of
perfect matchings; this paper would not exist without him. Thanks also to Kenji
Maillard, Michele Pagani, Marc Bagnol, Antoine Amarilli, Alexis Saurin, Stefano
Guerrini and Virgile Mogbil for discussions, references and encouragements, and
to Thomas Seiller for his writing advice on the initial conference version.

We are also grateful to the anonymous reviewers for their useful and detailed
feedback on previous versions of this paper.


\bibliographystyle{alpha}
\bibliography{/home/tito/export.bib}

\appendix

\section{Proof of Lemma~\ref{lemma-prescribed}}%
\label{appendix-proof-prescribed}

We rely on a version of Berge's lemma (Lemma~\ref{berge}) for paths:
\begin{lem}[Berge~\cite{berge_two_1957}]%
  \label{berge-augmenting}
  Let $G$ be a graph and $M$ be a matching of $G$. If $P$ is an \emph{augmenting
    path} for $M$ --- i.e., an alternating path whose endpoints are unmatched ---
  then $M \triangle P$ is a matching and $|M \triangle P| = |M| + 1$. (Thus,
  adding $P$ \enquote{augments} $M$, hence the name.) Conversely, if $M$ is a
  matching with $|M'| > |M|$, then $M \triangle M'$ is a vertex-disjoint union
  of:
  \begin{itemize}
  \item $|M'| - |M|$ augmenting paths for $M$;
  \item some (possibly zero) cycles which are alternating for both $M$ and $M'$.
  \end{itemize}
\end{lem}

\noindent
Let $u,v \in V$ be the unmatched vertices. If there is an augmenting path for
$M$ in $G$, its endpoints must be $u$ and $v$, and this is equivalent to the
existence of a perfect matching in $G$. Let $e = (a,b)$, $G' = (V, E \setminus
\{e\})$ and $M' = M \setminus \{e\}$.

Suppose $G'$ admits a perfect matching $M''$. Then the symmetric difference $M'
\triangle M''$ consists of two vertex-disjoint alternating paths for $M'$ whose
endpoints are $\{u, v, a, b\}$, by Berge's lemma for paths; indeed, our
assumptions prevent the existence of alternating cycles for $M$, and therefore
for $M' \subset M$ as well.

We claim that these paths either go from $u$ to $a$ and $b$ to $v$, or from
$u$ to $b$ and $a$ to $v$. Otherwise, there would be an alternating path from
$a$ to $b$ for $M'$ in $G'$, and together with $(a,b) = e \in M$, this would
give us an alternating cycle for $M$ in $G$.

In both cases, let us join the two paths together by adding $e$. We get a path
starting with $u$, ending with $v$, crossing $e$ and alternating for $M$ in
$G$. Conversely, from such a path, one can get a perfect matching in $G'$.

It is clear that the reduction is in $\mathsf{AC}^0$. For the linear time
complexity, we exploit the fact that we already have at our disposal a matching
$M'$ of $G'$ which leaves only 4 vertices unmatched. A perfect matching can then
be found as follows: find a first augmenting path $P$ for $M'$ in linear time,
and then a second one $P'$ for $M' \triangle P$, both steps being done in linear
time (using a similar (but simpler) algorithm than for \textsc{UniquenessPM},
see~\cite{gabow_linear-time_1985} and {\cite[Section~9.4]{tarjan_data_1983}}).
If both augmenting paths exist, then $M \triangle P \triangle P'$ is a perfect
matching, and conversely, if $G'$ admits a perfect matching, then the procedure
succeeds in finding some $P$ and $P'$. (This does not mean that $P$ and $P'$ are
the same as the paths in the previous part of the proof, since they may not be
vertex-disjoint.)

\end{document}